\documentclass{amsart}
\allowdisplaybreaks[4]
\usepackage{graphicx}
\usepackage{color}
\usepackage{bbm}

\newcommand{\sig}{\sigma}

\newcommand{{\ba}}{\bf a}
\newcommand{\ve}{\varepsilon}
\newcommand{\la}{\lambda}

\newcommand{\ga}{\gamma}

\newcommand{\pa}{\partial}
\newcommand{\ra}{\rightarrow}

\newcommand{\del}{\delta}

\newcommand{\na}{\nabla}

\newcommand{\De}{\Delta}
\newcommand{\al}{\alpha}

\newcommand{\be}{\begin{equation}}
\newcommand{\ee}{\end{equation}}

\newcommand{\om}{\omega}

\newtheorem{lem}{Lemma}{\bf}{\it}
{\it}{\rm}
\newtheorem{rem}{Remark}{\it}{\rm}
{\it}{\rm}

\newtheorem{theorem}{Theorem}
\newtheorem{proposition}{Proposition}

\numberwithin{theorem}{section}
\numberwithin{lem}{section}
\numberwithin{rem}{section}
\numberwithin{equation}{section}
\numberwithin{proposition}{section}
\numberwithin{corollary}{section}
\title{ Extensions of the Brascamp-Lieb Inequality and the Dipole Gas}

\author{Joseph G. Conlon  and Michael Dabkowski}

\address{(Joseph G. Conlon): University of Michigan\\ Department of Mathematics\\ Ann Arbor,
  MI 48109-1109}
\email{conlon@umich.edu}

\address{(Michael Dabkowski): University of Michigan-Dearborn \\ Department of Mathematics and Statistics\\ Dearborn,
  MI 48128}
\email{mgdabkow@umich.edu}


\begin{document}

\maketitle

\begin{abstract}
This paper is concerned with $d\ge2$ lattice field models with action $V(\na\phi(\cdot))$, where $V:\mathbb{R}^d\ra\mathbb{R}$ is a uniformly convex function.  The main result Theorem 1.4 proves that charge-charge correlations in the Coulomb dipole gas are close to Gaussian. These results go beyond previous results of Dimock-Hurd and Conlon-Spencer. The approach in the paper is based on the observation that the sine-Gordon probability measure corresponding to  the dipole gas is the invariant measure for a certain stochastic dynamics.  The stochastic dynamics here differs from the stochastic dynamics in previous work used  to study the problem.  
\end{abstract}

\section{Introduction}
Let $X$ be a random variable taking values in $\mathbb{R}^n$. The pdf of $X$ satisfies a Poincar\'{e} inequality if the variance of $C^1$ functions of $X$ are bounded by a constant times the expectation of the square of the Euclidean norm of the gradient of the function, i.e. there exists $m^2>0$ such that
\be \label{A1}
{\rm Var}[F(X)] \ \le \ \frac{1}{m^2} \langle  \|DF(X)\|_2^2    \rangle \quad {\rm for \ all \ } C^1 \ {\rm functions \ } \   \ F:\mathbb{R}^n\ra\mathbb{R} \ ,
\ee
where $DF(\cdot)$ denotes the gradient of $F(\cdot)$. Suppose the probability density function for $X$ is proportional to the function $\phi\ra\exp[-W(\phi)/\ve], \ \phi\in\mathbb{R}^n$, where $W:\mathbb{R}^n\ra\mathbb{R}$ is a $C^2$ convex function and $\ve>0$.  Denoting by $\langle\cdot\rangle_\ve$ expectation with respect to this pdf, the Brascamp-Lieb (BL) inequality \cite{bl} implies that
\be \label{B1}
\ve^{-1}{\rm Var}_\ve[F(X)] \ \le \  \langle  [DF(X), \{D^2W(X)\}^{-1} DF(X)]_n \ \rangle_\ve \  ,
\ee
where $D^2W(\cdot)$ is the Hessian of $W(\cdot)$ and $[\cdot,\cdot]_n$ denotes the Euclidean inner product on $\mathbb{R}^n$.  If $D^2W(\cdot)\ge m^2 I_n>0$ in the quadratic form sense then (\ref{B1}) with $\ve=1$ implies (\ref{A1}). 

There are many proofs of the BL inequality (\ref{B1}).  One that is particularly relevant to this paper is a simple consequence of the Helffer-Sj\"{o}strand (HS) representation \cite{hs} for the variance of $F(X)$. This is a formula for the variance of $F(X)$ as the expectation of a quadratic form in $DF(X)$.  In  \cite{ns1} Naddaf and Spencer  used the HS representation \cite{hs} and techniques of homogenization of elliptic PDE to prove convergence of the large scale behavior of covariances of fields for the lattice Coulomb dipole gas in the sine-Gordon  representation to a Gaussian limit. In addition, in \cite{ns2} they combined the BL inequality with the classical Meyers' theorem \cite{m} to obtain a rate of convergence in homogenization of elliptic PDE with random coefficients.  The use of Meyers' theorem (which is a consequence of the Calder\'{o}n-Zygmund (CZ) theorem \cite{stein}) is now standard in the theory of homogenization \cite{jos,josotto}. The CZ theorem was used in \cite{cf} to obtain rates of convergence to the scaling limit  for the field covariances studied in \cite{ns1}. 

To define the functions $W(\cdot)$ associated with the dipole gas we  consider a periodic lattice $Q_L$  in $\mathbb{Z}^d$ centered at the origin with side of length an even integer $L$.  Hence points on $\mathbb{Z}^d\cap \pa Q_L$ are identified, whence the number of distinct lattice points in $Q_L$ is $|Q_L|=L^d$.    We define the gradient operator $\na$ on fields $\phi:Q_L\ra\mathbb{R}$ by
\be \label{C1}
\na\phi(x) \ =  \ [\na_1\phi(x),\dots,\na\phi_d(x))] \ , \quad \na_i\phi(x) \ = \ \phi(x+\mathbf{e}_i)-\phi(x), \ 
\ee
where the vector $\mathbf{e}_i\in\mathbb{Z}^d$ has $1$ as the $i$th coordinate and $0$ for the other coordinates,  $1\le i\le d$. Here we always consider $\na$ to be  a $d$ dimensional {\it column} operator, with adjoint $\na^*$ which is  a $d$ dimensional {\it row} operator. Let $V:\mathbb{R}^d\ra\mathbb{R}$ be a  $C^2$ uniformly convex function, which satisfies 
\be \label{D1}
0<\la I_d \ \le \ V''(\cdot) \ \le \ I_d \ , \quad V(0)=0, \ \ V(\om)=V(-\om) \ , \ \  \om\in\mathbb{R}^d \ ,
\ee
for some  $\la, \ 0<\la\le 1$. In (\ref{D1}) the inequality on the Hessian $V''(\cdot)$ of $V(\cdot)$ is in the quadratic form sense. The function $W(\cdot)=W_{m,h,L}(\cdot)$ is defined in terms of $V(\cdot)$, depends on a parameter $m>0$, and a function $h:Q_L\ra\mathbb{R}^d$. We set
\be \label{E1}
W_{m,h,L}(\phi) \ = \ \sum_{x\in Q_L} \frac{m^2\phi(x)^2}{2}+ h(x)\cdot\nabla\phi(x)+V(\na \phi(x)) \ .
\ee
The Hessian of the function (\ref{E1}) is a self-adjoint  linear operator on  $\ell_2(Q_L,\mathbb{R})$ and is given by the formula
\be \label{F1}
D^2W_{m,h,L}(\phi)f(x) \ = \  [\na^*V''(\na \phi(x))\na+m^2] f(x) \quad f(\cdot)\in\ell_2(Q_L,\mathbb{R}), \ x\in Q_L \ .
\ee
It follows from (\ref{D1}), (\ref{F1})  that $W_{m,h,L}(\cdot)$ is a convex function on $\mathbb{R}^{L^d}$.   Let $\langle\cdot\rangle_{\ve,m,h,L}$ denote expectation with respect to the probability measure with density proportional to the function $\phi\ra\exp\left[     -W_{m,h,L}(\phi)/\ve         \right], \ \phi\in\mathbb{R}^{L^d}$. We denote the inner product and norm on the Euclidean space 
$\ell_2(Q_L,\mathbb{R}^d)$  by $[\cdot,\cdot]_L$ and $\|\cdot\|_{2,L}$ respectively. The  BL inequality  (\ref{B1}) then  implies that
\begin{multline} \label{G1}
\ve^{-1}{\rm var}_{\ve,m,h,L}\left\{  [a(\cdot),\na\phi(\cdot)]_L \ \right\} \\ 
\le \  [\na^*a(\cdot), \{-\la\De+m^2\}^{-1}\na^*a(\cdot) \ ]_L \ \le \ \la^{-1}\|a\|_{2,L}^2 \ , \quad a\in\ell_2(Q_L,\mathbb{R}^d) \ ,
\end{multline}
where $-\De=\na^*\na$ is the negative discrete Laplacian. 
Note that we may take $L\ra\infty$ and $m\ra0$ in (\ref{G1}) to obtain a bound on functions $a\in\ell_2(\mathbb{Z}^d,\mathbb{R}^d)$.

For $\ve,m>0$ we shall be interested in properties of the function $q_{\ve,m,L}(\cdot)$, with  the space of functions $h:Q_L\ra\mathbb{R}^d$ as domain, defined by
\be \label{H1}
q_{\ve,m,L}(h(\cdot)) \ = \ -\ve\log\int_{\mathbb{R}^{L^d}}d\phi(\cdot) \ \exp\left[   -\frac{W_{m,h,L}(\phi) }{\ve}  \right] \ .
\ee
Evidently we have that
\be \label{I1}
q_{\ve,m,L}(h(\cdot))-q_{\ve,m,L}(0) \ = \ -\ve\log\ \ \left\langle \  \exp\left[   -\frac{ [h(\cdot),\na\phi(\cdot)]_L }{\ve}  \right] \ \ \right\rangle_{\ve,m,0,L} \ .
\ee
In the Gaussian case the function (\ref{I1}) is quadratic in $h(\cdot)$. Let us suppose that $V(\cdot)$ is given by
\be \label{J1}
V(\om) \ = \ \frac{1}{2}\om^*A\om \ , \quad  A \ {\rm symmetric \ and \ \ } 0<\la I_d\le A\le I_d \ .
\ee
Then (\ref{I1}) yields the formula
\be \label{K1}
q_{\ve,m,L}(h(\cdot))-q_{\ve,m,L}(0) \ = \ -\frac{1}{2}\left[ h,\na(\na^*A\na+m^2)^{-1}\na^*h\  \right]_L \ , \quad h\in \ell_2(Q_L,\mathbb{R}^d) \ .
\ee
We easily see from (\ref{J1}), (\ref{K1}) that in the Gaussian case one has the inequality
\be \label{L1}
-\frac{1}{2}\left[ h,\na(-\la\De+m^2)^{-1}\na^*h\  \right]_L \ \le \ q_{\ve,m,L}(h(\cdot))-q_{\ve,m,L}(0) \ \le \ -\frac{1}{2}\left[ h,\na(-\De+m^2)^{-1}\na^*h\  \right]_L \ .
\ee
We may analytically continue the function (\ref{K1}) to complex $h\in \ell_2(Q_L,\mathbb{C}^d)$. Writing $h(\cdot)=\Re h(\cdot)+i\Im h(\cdot)$ with $\Re h(\cdot), \ \Im h(\cdot)\in \ell_2(Q_L,\mathbb{R}^d)$, we have from (\ref{K1}) that
\begin{multline} \label{M1}
\Re \left[ q_{\ve,m,L}(h(\cdot))-q_{\ve,m,L}(0) \ \right] \ = \\
 -\frac{1}{2}\left[ \Re h,\na(\na^*A\na+m^2)^{-1}\na^*\Re h\  \right]_L+ \frac{1}{2}\left[ \Im h,\na(\na^*A\na+m^2)^{-1}\na^*\Im h\  \right]_L \ , \quad h\in \ell_2(Q_L,\mathbb{C}^d) \ .
\end{multline}
It follows from (\ref{J1}), (\ref{M1}) that
\begin{multline} \label{N1}
-\frac{1}{2}\left[\Re h,\na(-\la\De+m^2)^{-1}\na^*\Re h\  \right]_L + \frac{1}{2}\left[\Im h,\na(-\De+m^2)^{-1}\na^*\Im h\  \right]_L \ \le \ \Re[q_{\ve,m,L}(h(\cdot))-q_{\ve,m,L}(0)]  \\
 \le \ -\frac{1}{2}\left[\Re h,\na(-\De+m^2)^{-1}\na^*\Re h\  \right]_L + \frac{1}{2}\left[\Im h,\na(-\la\De+m^2)^{-1}\na^*\Im h\  \right]_L \ .
\end{multline}

In this paper we shall be concerned with showing that the inequalities (\ref{L1}), (\ref{N1}),  which hold for the Gaussian case (\ref{J1}),  extend to  certain non-Gaussian $V(\cdot)$. In the case of the inequality (\ref{L1}) for real $h(\cdot)$ we require that $V(\cdot)$ satisfies (\ref{D1}).  In the case of complex $h(\cdot)$ the inequality (\ref{N1}) may be extended to non-Gaussian $V(\cdot)$ provided $V(\cdot)$ satisfies (\ref{D1}) and  can be analytically continued in a uniform way to a strip parallel to the real axis. More precisely, letting $|\cdot|$ denote the Euclidean norm on $\mathbb{C}^d$, we assume that the function $\om\ra V''(\om)$ is holomorphic and uniformly continuous in a strip, so for every $\eta>0$ there exists $\del(\eta)>0$ such that 
\be \label{O1}
\|V''(\om)-V''(\Re\om)\| \ < \ \eta \ , \quad {\rm for \ } \om=\Re\om+i\Im\om\in\mathbb{C}^d, \ \ |\Im\om|<\del(\eta) \ ,
\ee
where $\|\cdot\|$ in (\ref{O1}) denotes the Euclidean matrix norm of $d\times d$ complex matrices.

Inequalities of the type (\ref{L1}), (\ref{N1}) for non-Gaussian $V(\cdot)$ have already been proven in \cite{bl}, \cite{fs81}, which we summarize as follows:
\begin{theorem}
Assume  $V(\cdot)$ satisfies (\ref{D1}). Then the function $h(\cdot)\ra q_{\ve,m,L}(h(\cdot))$ with domain all real functions $h:Q_L\ra\mathbb{R}^d$, is concave and the  inequality (\ref{L1}) holds.  Suppose in addition that  $V(\cdot)$ satisfies(\ref{O1}). Then the lower bound of (\ref{N1}) for $h:Q_L\ra\mathbb{C}^d$ approximately holds, precisely:
\begin{multline} \label{P1}
\Re[q_{\ve,m,L}(h(\cdot))-q_{\ve,m,L}(0)] \ \ge \\
  -\frac{1}{2}\left[\Re h,\na(-\la\De+m^2)^{-1}\na^*\Re h\  \right]_L + \frac{1}{2}\left[\Im h,\na[-(1+\eta)\De+m^2]^{-1}\na^*\Im h\  \right]_L  \ ,
\end{multline}
provided the norm of $\Im h(\cdot)$ in $\ell_p(Q_L,\mathbb{R}^d)$  satisfies for some $p$ with $1<p<\infty$ the inequality
\be \label{Q1}
\|\Im h(\cdot)\|_{p,L} \ \le  \  \frac{\del(\eta)}{\kappa_p} \ ,
\ee
where $\kappa_p$ is a constant  independent of $L,m$ as $L\ra\infty, \ m\ra0$. Furthermore, one has that $\lim_{p\ra 2}\kappa_p=1, \  \lim_{p\ra\infty}\kappa_p=\infty$. 
\end{theorem}  
\begin{rem}
Note that the function $p\ra \|\Im h(\cdot)\|_{p,L}, \ 1< p<\infty, $ is decreasing. The constant $\kappa_p$ appears in the CZ theorem, Theorem 2.1.
\end{rem}

To go beyond the results of Theorem 1.1 we use the fact that the distribution of the random variable $X$ in the BL inequality (\ref{B1}) is the invariant measure for the solution 
to a stochastic differential equation (SDE). The probability measure with density proportional to $\phi\ra\exp[-W(\phi)/\ve], \ \phi\in\mathbb{R}^n,$ is invariant for the stochastic dynamics
\be \label{R1}
d\phi(t) = b(\phi(t)) \ dt+\sqrt{\ve} \  dB(t) \ , \quad {\rm where \ } b(\phi)=-\frac{1}{2}DW(\phi) \ ,
\ee
and $B(\cdot)$ denotes $n$ dimensional Brownian motion.
The condition for  a density $\phi\ra u(\phi), \ \phi\in\mathbb{R}^n,$ to be invariant for the dynamics (\ref{R1}) is that it satisfy the elliptic equation
\be \label{S1}
D^*[-b(\phi)u(\phi)+\frac{\ve}{2} Du(\phi)] \ = 0 \ , \quad \phi\in\mathbb{R}^n \ ,
\ee
where $D^*$ denotes the divergence operator. Note that the function $u(\phi)=\exp[-W(\phi)/\ve]$ is a solution to the equation
\be \label{T1}
-b(\phi)u(\phi)+\frac{\ve}{2} Du(\phi) \ = 0 \ , \quad \phi\in\mathbb{R}^n \ ,
\ee
and hence (\ref{S1}).  It follows from (\ref{T1}) that for any  $n\times n$ symmetric positive definite matrix $A$ this function $\phi\ra u(\phi)$ is also a solution to the equation
\be \label{U1}
D^*[-Ab(\phi)u(\phi)+\frac{\ve}{2} ADu(\phi)] \ = 0 \ , \quad \phi\in\mathbb{R}^n \ .
\ee
Since the function $\phi\ra u(\phi)$ is a solution to (\ref{U1}), it  is also an invariant density for the stochastic dynamics
\be \label{V1}
d\phi(t) = Ab(\phi(t)) \ dt+\sqrt{\ve A} \  dB(t) \ .
\ee

In $\S3$ we shall apply the approach of the previous paragraph to the function $W(\cdot)$ defined by (\ref{E1})  to obtain an alternative proof of Theorem 1.1 and go beyond it.
\begin{theorem}
Assume  $V(\cdot)$ satisfies (\ref{D1}), (\ref{O1}). Then for $\eta$ in the interval $0<\eta<\la$ the function $h(\cdot)\ra q_{\ve,m,L}(h(\cdot))$ with domain all real functions $h:Q_L\ra\mathbb{R}^d$, can be analytically continued to the domain of complex valued functions $h:Q_L\ra\mathbb{C}^d$ in $\ell_2(Q_L,\mathbb{C}^d)$ satisfying $\|\Im h(\cdot)\|_{2,L}< (\la-\eta)\del(\eta)$.  Furthermore, in this domain the function $q_{\ve,m,L}(\cdot)$ satisfies the upper bound
\begin{multline} \label{W1}
 \Re[q_{\ve,m,L}(h(\cdot))-q_{\ve,m,L}(0)] \\
  \le \ -\frac{1}{2}\left[\Re h,\na(-\De+m^2)^{-1}\na^*\Re h\  \right]_L + \frac{1}{2(\la-\eta)}\left[\Im h,\na(-\De+m^2)^{-1}\na^*\Im h\  \right]_L \ ,
\end{multline}
and the lower bound,
\begin{multline} \label{X1}
 \Re[q_{\ve,m,L}(h(\cdot))-q_{\ve,m,L}(0)] \\
  \ge \ -\frac{1}{2\la}\left[\Re h,\na(-\De+m^2)^{-1}\na^*\Re h\  \right]_L + \frac{1}{2}\left[1-\frac{\eta}{(\la-\eta)^2}\right]\left[\Im h,\na(-\De+m^2)^{-1}\na^*\Im h\  \right]_L \ .
\end{multline}
\end{theorem}  
\begin{rem}
The inequality (\ref{X1}) is of the same type as (\ref{P1}) but is weaker since (\ref{X1}) requires  $\|\Im h(\cdot)\|_{2,L}< (\la-\eta)\del(\eta)$ whereas (\ref{P1}) just requires
 $\|\Im h(\cdot)\|_{2,L}< \del(\eta)$. In addition, the term involving $\Im h(\cdot)$ on the RHS of (\ref{X1}) depends on $\la\le 1$ but the corresponding bound in (\ref{P1}) is independent of $\la$. The inequality (\ref{X1}) is obtained as a consequence of the bound (\ref{DR3}) on a second derivative  of the function $h(\cdot)\ra q_{\ve,m,L}(h(\cdot))$. The methodology used to prove Theorem 1.1 can only be applied to obtain bounds on second derivatives  of the function $h(\cdot)\ra q_{\ve,m,L}(h(\cdot))$ when $h(\cdot)$ is real.
\end{rem}

We assume now that the function $V:\mathbb{R}^d\ra\mathbb{R}$, in addition to satisfying (\ref{D1}) is also $C^3$. Defining $\|V'''(\om)\|$ by 
\be \label{Y1}
\|V'''(\om)\| \ = \ \sup\{ |V'''(\om)[\xi_1,\xi_2,\xi_3]|: \ \xi_j\in\mathbb{C}^d, \ |\xi_j|=1, \ j=1,2,3\} \ ,
\ee
let $M>0$ have the property
\be \label{Z1}
\sup_{\om\in\mathbb{R}^d} \|V'''(\om)\|  \ \le \ M \ < \ \infty \ .
\ee
It was shown in \cite{cs} that if (\ref{D1}), (\ref{Z1}) hold then for $h(\cdot)\in\ell_2(Q_L,\mathbb{R}^d)$  one has the inequality
\be \label{AA1}
\left| \ q_{\ve,m,L}(h(\cdot))-q_{\ve,m,L}(0)+\frac{1}{2\ve}\left\langle \ [h(\cdot),\na\phi(\cdot)]_L^2 \ \right\rangle_{\ve,m,0,L} \ \right| \ \le \ C(\la)M\|h\|_{2,L}^3 \ ,
\ee
where $C(\la)=1/6\la^2(\la-1/2)$, provided $\la$ in (\ref{D1}) satisfies the inequality $\la>1/2$.  In $\S4$ we show that the inequality (\ref{AA1}) holds for $0<\la\le 1$ with constant $C(\la)=1/6\la^3$.

We can also extend the inequality (\ref{AA1}) to complex $h(\cdot)$  in $\ell_p(Q_L,\mathbb{C}^d)$ as was the case in Theorem 1.1. For $V(\cdot)$ holomorphic in a strip  and  satisfying  (\ref{D1}), (\ref{O1}), (\ref{Z1}) there exists $M_\eta>0$ such that 
\be \label{AB1}
\sup_{\om\in\mathbb{C}^d: \ \|\Im\om\|<\del(\eta)} \|V'''(\om)\|  \ \le \ M _\eta\ < \ \infty \ .
\ee
\begin{theorem}
Assume  $V(\cdot)$ satisfies (\ref{D1}), (\ref{O1}), (\ref{AB1}) and $p$ is in the interval $1<p\le 3$. Then for $\eta$ satisfying $0<\eta< \la-(1-1/\kappa_p)$ and $h\in\ell_p(Q_L,\mathbb{C}^d)$ satisfying
$\|\Im h(\cdot)\|_{p,L}< [(\la-\eta)-(1-1/\kappa_p)]\del(\eta)$ there is the inequality
\be \label{AC1}
\left| \ q_{\ve,m,L}(h(\cdot))-q_{\ve,m,L}(0)+\frac{1}{2\ve}\left\langle \ [h(\cdot),\na\phi(\cdot)]_L^2 \ \right\rangle_{\ve,m,0,L} \ \right| \ \le \ 
\frac{M_\eta \|h\|_{p,L}^3}{6[(\la-\eta)-(1-1/\kappa_p)]^3} \ .
\ee
\end{theorem}
We  may interpret Theorems 1.1-1.3 using (\ref{I1}) as results about expectations of certain functions of gradient fields on $Q_L$. These can be extended to results about gradients fields on $\mathbb{Z}^d$.
We define the measure $\langle\cdot\rangle_\ve$ on $\mathbb{Z}^d$ gradient fields $\om(x)=\na\phi(x)\in\mathbb{R}^d, \ x\in\mathbb{Z}^d,$  by  the limits
\be \label{AD1}
\left\langle \ f\left(\om(x_1),\dots,\om(x_N)\right) \ \right\rangle_\ve \ = \ \lim_{m\ra 0}\lim_{L\ra\infty}  \left\langle \    f\left(\na\phi(x_1),\dots,\na\phi(x_N) \right)   \   \right\rangle_{\ve,m,0,L} \ .
\ee
One can see using the BL inequality \cite{conlon1} that the limit (\ref{AD1}) exists for  a wide class of functions $f(\cdot)$, including the function on the RHS of (\ref{I1}).   The measure $\langle \cdot\rangle_\ve$ in (\ref{AD1})  was first constructed in \cite{fs} (see also \cite{arms}). 

We show how the previous theorems can be applied to estimate  expectations and covariances of the variables $\phi(\cdot)\ra\exp\left[\rho\phi(x)/\ve\right], \ x\in\mathbb{Z}^d, \ \rho\in\mathbb{C},$ with respect to the measure $\langle\cdot\rangle_\ve$.  For $\nu>0$ we denote by $G_\nu(\cdot)$ the Green's function on $\mathbb{Z}^d$, which is the solution to the discrete Helmholtz equation
\be \label{AE1}
[\na^*\na+\nu]G_\nu(y)   \ = \ \del(y), \ y\in\mathbb{Z}^d \ ,
\ee
where $\del(\cdot)$ is the Kronecker delta function.  Letting $h_{x,\nu}:\mathbb{Z}^d\ra\mathbb{R}^d$ be the function $h_{x,\nu}(y)=\na G_\nu(y-x), \ y\in\mathbb{Z}^d,$ and denoting by $[\cdot,\cdot]_\infty$ the Euclidean inner product on $\ell_2(\mathbb{Z}^d,\mathbb{R}^d)$, it follows from (\ref{AE1}) that
\be \label{AF1}
[h_{x,\nu}(\cdot),\na\phi(\cdot)]_\infty \ = \  \phi(x)-\nu[G_\nu(\cdot-x),\phi(\cdot)]_\infty \ .
\ee
The identity (\ref{AF1}) holds for all $\na\phi\in\ell_2(\mathbb{Z}^d,\mathbb{R}^d)$ and $\nu>0$.  If $d\ge 3$ we may take the limit as $\nu\ra0$ in (\ref{AF1}) to obtain the identity 
$[h_{x,0}(\cdot),\na\phi(\cdot)]_\infty \ = \  \phi(x)$. We have however for all $d\ge2$ that
\be \label{AG1}
\lim_{\nu\ra0} [h_{x,\nu}(\cdot)-h_{0,\nu}(\cdot),\na\phi(\cdot)]_\infty  \ = \ \phi(x)-\phi(0) \ , \quad x\in\mathbb{Z}^d \ .
\ee
The expectation
\be \label{AH1}
\left\langle     \exp\left[     -\frac{\rho \phi(x)}{\ve}         \right]            \right\rangle_\ve \ = \ \left\langle     \exp\left[     -\frac{\rho \phi(0)}{\ve}         \right]            \right\rangle_\ve  \ , \quad  x\in\mathbb{Z}^d \ ,
\ee
may be estimated using (\ref{I1}) and  theorems 1.1, 1.2.  For $\rho\in\mathbb{R}$ we obtain from theorem 1.1  the inequalities
\be \label{AI1}
\exp\left[    \frac{C_d\rho^2}{2\ve}         \right] \ \le \ \left\langle     \exp\left[     -\frac{\rho \phi(0)}{\ve}         \right]            \right\rangle_\ve  \ \le \ \exp\left[    \frac{C_d\rho^2}{2\la\ve}         \right] \  , \quad \rho\in\mathbb{R} \ .
\ee
where the positive constant $C_d$ depends only on $d$. If $d\ge 3$ then $C_d$ is finite but $C_2=\infty$.  If $\rho$ is pure imaginary the expectation (\ref{AH1}) is real  by the reflection invariance (\ref{D1}). We then have from theorems 1.1, 1.2 the inequalities
\begin{multline} \label{AJ1}
\exp\left[   - \frac{C_d\mu^2}{2(\la-\eta)\ve}         \right] \ \le \ \left\langle     \exp\left[     -\frac{i\mu \phi(0)}{\ve}         \right]            \right\rangle_\ve  \ \le \ \exp\left[   - \frac{C_d\mu^2}{2(1+\eta)\ve}         \right] \  , \\
 \mu\in\mathbb{R} \ , \quad |\mu| \ <c_d(\la-\eta)\del(\eta) \ ,
\end{multline}
where the constants $C_d,c_d$ depend only on $d$ and $C_d$ is the same constant as in (\ref{AI1}).  If $d\ge 3$ the constants $C_d,c_d$ are finite and positive, but for $d=2$ there is a singularity, as has already been observed. To understand the nature of the singularity we replace  the expectation in (\ref{AJ1}) by the expectation in which $\phi(0)$ is replaced by  $[h_{0,\nu}(\cdot),\na\phi(\cdot)]_\infty $ and then take $\nu\ra0$. From theorem 1.1 the upper bound in (\ref{AJ1}) holds  provided $|\mu|\|h_{0,\nu}(\cdot)\|_p\le \del(\eta)/\kappa_p$.  Since $|h_{0,\nu}(y)|\le C/[1+|y|]^{d-1}, \ y\in\mathbb{Z}^d$,  it follows that in $d=2$ the function $h_{0,\nu}(\cdot)$ is uniformly summable in $\ell_p(\mathbb{Z}^d,\mathbb{R}^d)$ as $\nu\ra0$ for any $p>2$.  The upper bound in (\ref{AJ1}) therefore holds in an interval $|\mu|<c_2\del(\eta)$  for some constant $c_2>0$, uniformly as $\nu\ra0$.  Since the corresponding  lower bound on the RHS of (\ref{P1}) converges to $+\infty$ as $\nu\ra0$ we conclude that the upper bound in (\ref{AJ1}) converges to $0$. 

The covariance of the variables $\exp\left[\rho\phi(x)/\ve\right], \ x\in\mathbb{Z}^d,$ with $\exp\left[-\rho\phi(0)/\ve\right]$ may be estimated  for large $|x|$  in terms of the Green's function for the homogenized constant coefficient elliptic PDE
\be \label{AK1}
-\na\mathbf{a}_{\ve,{\rm hom}}\na u_{\ve,{\rm hom}}(x) \ =  \ f(x) \ , \quad x\in\mathbb{R}^d \ ,
\ee
associated with the massless measure $\langle\cdot\rangle_\ve$, which was obtained by Naddaf and Spencer \cite{ns1}.  The $d\times d$ matrix $\mathbf{a}_{\ve,{\rm hom}}$ in (\ref{AK1}) is symmetric positive definite and satisfies the quadratic form inequalities
\be \label{AL1}
0 \ < \ \la I_d \ \le \ \mathbf{a}_{\ve,{\rm hom}} \ \le \ I_d \ .
\ee
The inequality (\ref{AL1}) is a consequence of (\ref{D1}).  Denoting by $G_{\mathbf{a}_{\ve,{\rm hom}}}(x) , \ x\in\mathbb{R}^d,$ the Green's function for (\ref{AK1}) we note that at large $|x|$,
\be \label{AM1}
\begin{array}{ccc}
G_{\mathbf{a}_{\ve,{\rm hom}}}(x)  \ &\simeq&  \ |x|^{-(d-2)}  \ , \quad d\ge 3 \ ,\\
G_{\mathbf{a}_{\ve,{\rm hom}}}(0)-G_{\mathbf{a}_{\ve,{\rm hom}}}(x)  \ &\simeq&  \quad \log |x| \quad d=2 \ .
\end{array}
\ee
\begin{theorem}
Assume  $V(\cdot)$ satisfies (\ref{D1}), (\ref{O1}), (\ref{AB1}). Then there is a constant $c_d>0$, depending only on $d\ge2$ such that for $\rho\in\mathbb{C}$ satisfying $|\Im\rho|<c_d(\la-\eta)\del(\eta)$ the following hold: If $d\ge 3$ and $1-\la,\eta$ are sufficiently small, then one has that
\begin{multline} \label{AN1}
{\rm cov}_\ve\left\{   \exp\left[\rho\phi(x)/\ve\right], \        \exp\left[-\rho\phi(0)/\ve\right]                \right\} \ = \\
\left\langle     \exp\left[     -\frac{\rho \phi(0)}{\ve}         \right]            \right\rangle_\ve^2\left\{\exp\left[  \frac{-\rho^2 G_{\mathbf{a}_{\ve,{\rm hom}}}(x) +|\rho|^3{\rm Error}_\ve(x)}{\ve}  \right]-1\right\} \ ,
\end{multline}
where $|{\rm Error}_\ve(x)|\le C/[|x|^{d-2+\al}+1]$ for some $\al>0$. If $d=2$ and $p$ satisfies $2<p\le 3$, there is a constant $c_p$ such that if  
$|\Im\rho|<c_p[(\la-\eta)-(1-1/\kappa_p)]\del(\eta)$ then
\begin{multline} \label{AO1}
   \left\langle     \exp\left[     \frac{\rho \{\phi(x)-\phi(0)\}}{\ve}         \right]            \right\rangle_\ve    \ = \\
             \exp\left[\frac{\rho^2\{G_{\mathbf{a}_{\ve,{\rm hom}}}(0)-G_{\mathbf{a}_{\ve,{\rm hom}}}(x)\}+|\rho|^3{\rm Error}_\ve(x) }{\ve}\right] \ ,
\end{multline}
where $|{\rm Error}_\ve(x)|\le C$. The constants $C$ in (\ref{AN1}), (\ref{AO1}) are independent of $\ve>0$. 
\end{theorem}
\begin{rem}
The proof of  (\ref{AN1}) relies on estimates of  singular integral  operators on weighted $\ell_p$ spaces.  Estimates on weighted Hilbert spaces $\ell_2$ were already used in \cite{cs}.  Following the discussion after (\ref{AJ1}),  the expectation in (\ref{AO1}) when $\rho=i\mu$ with $\mu\in\mathbb{R}$  is also  formally a covariance.
\end{rem}

In the case of the dipole gas the function $V(\cdot)$ is given by the formula
\be \label{AP1}
V(\om) \ = \ \frac{|\om|^2}{2}-a\sum_{j=1}^d\cos\om_j \ , \quad \om=[\om_1,\dots,\om_d]\in\mathbb{R}^d \ .
\ee
If $a\in\mathbb{R}$ satisfies $|a|<1$ then the conditions  (\ref{D1}), (\ref{O1}), (\ref{AB1}) hold for $V(\cdot)$ in (\ref{AP1}), whence the results of Theorem 1.4 apply to the dipole gas. In this case Theorem 1.4 goes beyond similar results  for the dipole gas obtained by Dimock and Hurd- Theorem 3 of  \cite{dh}. However the results of Theorem 3 of \cite{dh} also hold for $a\in\mathbb{C}$, but with $|a|\ll 1$.  The renormalization group method used in \cite{dh} (see also \cite{dimock}) does not allow for reasonable estimates on the smallness of the parameter $a$ in (\ref{AP1})  or of $\rho$ in the statement of Theorem 1.4.  However it is a powerful method and can be applied to some probability measures which are not uniformly convex. The method was first introduced by Gawedzki and Kupiainen \cite{gk} and later refined by Brydges and Yau \cite{bryau} in an influential paper (see also \cite{brydges}). 
\vspace{.1in}
\section{Proof of Theorem 1.1}
We use the BL inequality (\ref{G1})  to obtain the lower bound in (\ref{L1}). Applying the fundamental theorem of calculus (FTC) to (\ref{H1}) and using the symmetry property of (\ref{D1}) we have that
\begin{multline} \label{A2}
q_{\ve,m,L}(h(\cdot))-q_{\ve,m,L}(0) \ = \ \int_0^1d\al \ (1-\al)\frac{d^2}{d\al^2} q_{\ve,m,L}(\al h(\cdot))  \\
= \ -\int_0^1 d\al \ (1-\al) \ \ve^{-1}{\rm var}_{\ve,m,\al h,L}\left\{  [h(\cdot),\na\phi(\cdot)]_L \ \right\} \   .
\end{multline}
To obtain the upper bound in (\ref{L1}) we use contour deformation and Jensen's inequality. 
We can for any $\psi\in \ell_2(Q_L,\mathbb{R})$ make the deformation $\phi(\cdot)\ra\phi(\cdot)+\psi(\cdot)$  in the integration (\ref{H1}), which yields the identity
\begin{multline} \label{B2}
 \left\langle \  \exp\left[   -\frac{ [h(\cdot),\na\phi(\cdot)]_L }{\ve}  \right] \ \ \right\rangle_{\ve,m,0,L} \ = \
   \exp\left[   -\frac{m^2\|\psi\|_{2,L}^2 +2[h(\cdot),\na \psi(\cdot)]_L }{2\ve}  \right]  \ \times \\
   \left\langle \ \exp\left[   -\frac{1}{\ve}\left\{ m^2[\psi(\cdot),\phi(\cdot)]_L+  [h(\cdot),\na\phi(\cdot)]_L  +\sum_{x\in Q_L} [V(\na\phi(x)+\na \psi(x))-V(\na\phi(x))] \ \right\} \right] \ \ \right\rangle_{\ve,m,0,L}  \ .
\end{multline}
It follows from (\ref{D1}) that
\be \label{C2}
\sum_{x\in Q_L} [V(\na\phi(x)+\na \psi(x))-V(\na\phi(x))] \le \sum_{x\in Q_L} V'(\na\phi(x))\na\psi(x) +\frac{1}{2}\|\na\psi(\cdot)\|_{2,L}^2 \ .
\ee
Substituting the inequality (\ref{C2}) into (\ref{B2}) and using Jensen's inequality combined with the symmetry condition of (\ref{D1}) we conclude that
\begin{multline} \label{D2}
q_{\ve,m,L}(h(\cdot))-q_{\ve,m,L}(0) \ \le \\
    \min_{\psi\in\ell_2(Q_L,\mathbb{R})}\left\{\frac{m^2}{2}\|\psi\|_{2,L}^2 +[h(\cdot),\na \psi(\cdot)]_L +        \frac{1}{2}\|\na\psi(\cdot)\|_{2,L}^2  \ \right\}      
    \  = \     -\frac{1}{2}\left[ h,\na(-\De+m^2)^{-1}\na^*h\  \right]_L  \ .                               
\end{multline}

To show concavity of the function  $h(\cdot)\ra q_{\ve,m,L}(h(\cdot))$ we use directional derivatives. The directional derivative of  $q_{\ve,m,L}(\cdot)$ in direction $a\in \ell_2(Q_L,\mathbb{R}^d)$ is given by the formula
\be \label{N2}
D_aq_{\ve,m,L}(h(\cdot)) \ = \ \lim_{\eta\ra 0} \eta^{-1}\left[ q_{\ve,m,L}(h(\cdot)+\eta a(\cdot)) -q_{\ve,m,L}(h(\cdot)) \right] \ = \ \langle [a(\cdot),\na\phi(\cdot)]_L \ \rangle_{\ve,m,h,L} \ .
\ee
Using a similar formula to (\ref{A2}), we have  from FTC and  the BL inequality  that
\be \label{O2}
\left|\langle \ [a(\cdot), \na\phi(\cdot)]_L \ \rangle_{\ve,m,h,L}\right| \ \le \  \|(-\la\De+m^2)^{-1/2}\na^*a\|_{2,L}\|(-\la\De+m^2)^{-1/2}\na^*h\|_{2,L} \ ,
\ee
whence we obtain a bound on $D_a q_{\ve,m,L}(\cdot)$. The second directional derivative of $q_{\ve,m,L}(\cdot)$  in directions $a_1,a_2\in \ell_2(Q_L,\mathbb{R}^d)$  is obtained by differentiating (\ref{N2}) and is given  by the formula
\be \label{P2}
D^2_{a_1,a_2}q_{\ve,m,L}(h(\cdot)) \ =  \ -\ve^{-1}\ {\rm cov}_{\ve,m,h,L}\left\{  [a_1(\cdot),\na\phi(\cdot)]_L, \  [a_2(\cdot),\na\phi(\cdot)]_L                    \right\} \ ,
\ee
where ${\rm cov}_{\ve,m,h,L}(X,Y)$ denotes the covariance of the variables $X,Y$ with respect to the measure $\langle\cdot\rangle_{\ve,m,h,L}$. The concavity of the function $h(\cdot)\ra q_{\ve,m,L}(h(\cdot))$ follows from (\ref{P2}). We also obtain from (\ref{P2}) and the BL inequality the bound
\be \label{Q2}
|D^2_{a_1,a_2}q_{\ve,m,L}(h(\cdot))| \ \le \    \|(-\la\De+m^2)^{-1/2}\na^*a_1\|_{2,L}\|(-\la\De+m^2)^{-1/2}\na^*a_2\|_{2,L} \ .
\ee

To prove (\ref{P1}) we assume that both (\ref{D1}), (\ref{O1}) hold and  write
\begin{multline} \label{E2}
 \left\langle \  \exp\left[   -\frac{ [h(\cdot),\na\phi(\cdot)]_L }{\ve}  \right] \ \ \right\rangle_{\ve,m,0,L} \\ 
 \hspace{.8in} \ = \  \left\langle \  \exp\left[   -\frac{ [\Re h(\cdot),\na\phi(\cdot)]_L }{\ve}  \right] \ \ \right\rangle_{\ve,m,0,L} \  \left\langle \  \exp\left[   -\frac{i [\Im h(\cdot),\na\phi(\cdot)]_L }{\ve}  \right] \ \ \right\rangle_{\ve,m,\Re h,L} \ .
\end{multline}
The first expectation on the RHS of (\ref{E2}) can be bounded above by the BL inequality.  To bound the second expectation in absolute value we
make  for any $\psi\in \ell_2(Q_L,\mathbb{R})$ the imaginary deformation $\phi(\cdot)\ra\phi(\cdot)+i\psi(\cdot)$  in the integration (\ref{H1}), which yields the identity
\begin{multline} \label{F2}
 \left\langle \  \exp\left[   -\frac{ i[\Im h(\cdot),\na\phi(\cdot)]_L }{\ve}  \right] \ \ \right\rangle_{\ve,m,\Re h,L} \ = \
   \exp\left[   \frac{m^2\|\psi\|_{2,L}^2 +2[\Im h(\cdot),\na \psi(\cdot)]_L -2i[\Re h(\cdot),\na \psi(\cdot)]_L}{2\ve}  \right]  \ \times \\
   \left\langle \ \exp\left[   -\frac{1}{\ve}\left\{ m^2i[\psi(\cdot),\phi(\cdot)]_L+ i [\Im h(\cdot),\na\phi(\cdot)]_L  +\sum_{x\in Q_L} [V(\na\phi(x)+i\na \psi(x))-V(\na\phi(x))] \ \right\} \right] \ \ \right\rangle_{\ve,m,\Re h,L}  \ .
\end{multline}
Observe now that
\be \label{G2}
\Im V'(\om) \ = \ \Im[V'(\om)-V'(\Re\om)] \ =  \Re\int_0^1d\al \ V''(\al\om+(1-\al)\Re\om) \  \Im\om \ .
\ee
It follows from (\ref{O1}), (\ref{G2}) that for any $\eta>0$ there exists $\del(\eta)>0$ such that
\be \label{H2}
|\Im V'(\om)-V''(\Re\om)\Im\om| \ \le \  \eta|\Im\om| \quad {\rm for \ }\om\in\mathbb{C}^d, \ |\Im\om|<\del(\eta) \ .
\ee
We have then from (\ref{D1}), (\ref{H2})  that
\begin{multline} \label{I2}
\Re \sum_{x\in Q_L} [V(\na\phi(x)+i\na \psi(x))-V(\na\phi(x))] = \ - \sum_{x\in Q_L} \int_0^1d\al \ \Im V'(\na\phi(x)+i\al\na\psi(x))\na\psi(x) \\
\ge \ -\frac{(1+\eta)}{2}\|\na\psi(\cdot)\|_{2,L}^2 \ , \quad {\rm provided \ } \sup_{x\in Q_L}|\na\psi(x)|<\del(\eta) \  .
\end{multline}
We conclude from (\ref{F2}), (\ref{I2}) that
\begin{multline} \label{J2}
\Re[q_{\ve,m,L}(h(\cdot))-q_{\ve,m,L}(\Re h) \ ] \ \ge \\
    -\min_{\psi:Q_L\ra\mathbb{R}:\ \sup_{x\in Q_L}|\na\psi(x)|<\del(\eta)  }\left\{\frac{m^2}{2}\|\psi\|_{2,L}^2 +[\Im h(\cdot),\na \psi(\cdot)]_L +        \frac{(1+\eta)}{2}\|\na\psi(\cdot)\|_{2,L}^2  \ \right\}       \ .                          
\end{multline}
The minimizer of the quadratic form in (\ref{J2}) is 
\be \label{K2}
\psi_{\rm min}(\cdot) \ = \ -[-(1+\eta)\De+m^2]^{-1}\na^*\Im h(\cdot)   \ .
\ee
We conclude from (\ref{J2}) that
\be \label{L2}
q_{\ve,m,L}(h(\cdot))-q_{\ve,m,L}(\Re h) \ \ge \ \frac{1}{2}\left[\Im h,\na\{-(1+\eta)\De+m^2\}^{-1}\na^*\Im h\  \right]_L \ ,
\ee
provided $\sup_{x\in Q_L}|\na\psi_{\rm min}(x)|<\del(\eta)$.

The proof of Theorem 1.1 follows from (\ref{L2}) and  the following version of the CZ theorem \cite{stein}:
\begin{theorem}
For $1<p<\infty$ the operator $\na(-\De+m^2)\na^*$ is bounded on $\ell_p(Q_L,\mathbb{R}^d)$, with operator norm satisfying an inequality
\be \label{R2}
\left\| \na(-\De+m^2)\na^* \ \right\|_{\ell_p(Q_L,\mathbb{R}^d)} \ \le \  \kappa_p \ ,
\ee
where $\kappa_p\ge1$ is independent of $L,m$ as $L\ra\infty, \ m\ra0$. Furthermore, one has that $\lim_{p\ra 2}\kappa_p=1$. 
\end{theorem}
Now we just use the fact that $\sup_{x\in Q_L}|\na\psi_{\rm min}(x)|\le \left\|\na\psi_{\rm min}(\cdot)\right\|_{\ell_p(Q_L,\mathbb{R}^d)}$ for all $p, \ 1<p<\infty$, whence (\ref{Q1}) implies that
$\sup_{x\in Q_L}|\na\psi_{\rm min}(x)|<\del(\eta)$.
\begin{rem}
The contour deformation method employed in (\ref{B2}) can also be used to obtain a lower bound on the variance (\ref{G1}):
\be \label{S2}
\ve^{-1}{\rm var}_{\ve,m,h,L}\left\{  [a(\cdot),\na\phi(\cdot)]_L \ \right\} \
\ge \  [\na^*a(\cdot), \{-\De+m^2\}^{-1}\na^*a(\cdot) \ ]_L  \ .
\ee
When $h(\cdot)\equiv 0$ in (\ref{S2})  this follows immediately from (\ref{D2}) upon replacing $h(\cdot)$ by $\nu a(\cdot), \ \nu>0$, then dividing (\ref{D2}) by $\nu^2$ and letting $\nu\ra0$.  For general $h(\cdot)\in\ell_2(Q_L,\mathbb{R}^d)$ one needs to use the identity
\be \label{T2}
 \left\langle \  \frac{\pa W_{m,h,L}(\phi(\cdot))}{\pa\phi(x)} \ \right\rangle_{\ve,m,h,L}  \ = \ 0 , \quad x\in Q_L \ ,
\ee
after application of the Jensen inequality. The upper bound in (\ref{L1}) evidently follows from (\ref{A2}), (\ref{S2}).  The upper and lower bounds (\ref{G1}), (\ref{S2}) are also  easy consequences of the HS representation \cite{hs} for the variance. 
\end{rem}

\vspace{.1in}
\section{Stochastic Dynamics-Hilbert Space Theory}
The dynamics (\ref{R1}) corresponding to  (\ref{E1}) is given by
\begin{multline} \label{A3}
d\phi_{\ve,m,h,L}(x,t) \ = \ -\frac{1}{2}\left\{\na^*h(x)+m^2\phi_{\ve,m,h,L}(x,t)+ \na^*V'(\na\phi_{\ve,m,h,L}(x,t)) \right\} \ dt \\
+\sqrt{\ve}  \ dB(x,t), \quad x\in Q_L \ , t>0,
\end{multline}
where $B(x,\cdot), \ x\in Q_L,$ are independent copies of Brownian motion. 
The SDE (\ref{A3}) has been extensively studied in  \cite{fs,gos}.  Assuming $V(\cdot)$ satisfies (\ref{D1})   we apply the operator $A=[-\De+m^2]^{-1}$ with $-\De=\na^*\na$ to (\ref{A3}). The corresponding SDE (\ref{V1}) is then given by  
\begin{multline} \label{B3}
d\phi_{\ve,m,h,L}(x,t) \ = \ -\frac{1}{2}\Big[\phi_{\ve,m,h,L}(x,t)+[-\De+m^2]^{-1}\na^*h(x) \\
+ [-\De+m^2]^{-1}\na^*\{V'(\na\phi_{\ve,m,h,L}(x,t))- \na\phi_{\ve,m,h,L}(x,t)\} \Big] \ dt \\
+\sqrt{\ve}   \ [-\De+m^2]^{-1/2}W(x,t)  \ dt, \quad x\in Q_L \ , t>0 \ ,
\end{multline}
where $W(x,t) \ dt=dB(x,t),  \ x\in Q_L, \ t>0,$ are independent copies of  white noise. Taking  $\om_{\ve,m,h,L}(\cdot,t)=\na\phi_{\ve,m,h,L}(\cdot,t)\in\mathbb{R}^d$ we have from (\ref{B3}) that
\begin{multline} \label{C3}
d\om_{\ve,m,h,L}(\cdot,t) \ = \ -\frac{1}{2}\Big[\om_{\ve,m,h,L}(\cdot,t)+\na[-\De+m^2]^{-1}\na^*h(\cdot)\\
+ \na[-\De+m^2]^{-1}\na^*\{V'(\om_{\ve,m,h,L}(\cdot,t))- \om_{\ve,m,h,L}(\cdot,t)\} \Big] \ dt \\
+\sqrt{\ve}  \tilde{\om}_L(\cdot,t) \ dt, \quad \tilde{\om}_L(\cdot,t)  \ = \   \ \na[-\De+m^2]^{-1/2}W(\cdot,t),  \  t>0 \ .
\end{multline}
We shall show for certain functions $F$ on $\ell_2(Q_L,\mathbb{R}^d)$ that
\be \label{D3}
\lim_{T\ra\infty} \langle \ F(\om_{\ve,m,h,L}(\cdot,T)) \ \rangle \ = \  \langle \ F(\na\phi(\cdot)) \ \rangle_{\ve,m,h,L} \ .
\ee
 In the $\ve\ra 0$ deterministic case and $V''(\cdot)$ constant Gaussian case the limit is uniform as $L\ra\infty$ and $m\ra0$.  

The lower bound (\ref{L1}) has a more general form, which we may write as
\begin{multline} \label{E3}
\left\langle \ \exp\left[\frac{1}{\sqrt{\ve}}\left\{[a(\cdot),\na\phi(\cdot)]_L-\langle \ [a(\cdot),\na\phi(\cdot)]_L \ \rangle_{\ve,m,h,L} \right\} \ \right] \ \right\rangle_{\ve,m,h,L} \\
 \le \  \exp\left[  \frac{\|(-\la\De+m^2)^{-1/2}\na^*a\|_{2,L}^2}{2}                \right] \ , \quad a\in \ell_2(Q_L,\mathbb{R}^d) \ .
\end{multline}
We shall generalize the inequalities (\ref{G1}), (\ref{L1}) in the  form (\ref{E3}), and (\ref{O2})  to solutions $t\ra \om_{\ve,m,h,L}(\cdot,t)$ of (\ref{C3}).  To do this we use the Poincar\'{e} inequality for functions of white noise $W(x,t), \ x\in Q_L, \ 0<t<T$:
\be \label{F3}
{\rm Var} \left[ F(W(\cdot,t): \ 0<t<T)\right] \ \le \  \langle  \|D_{{\rm Mal}}F(W(\cdot,\cdot))\|^2_2         \rangle \ ,
\ee
where $D_{{\rm Mal}}F(W(\cdot,\cdot))\in L^2([0,T]\times Q_L,\mathbb{R})$ denotes the Malliavin derivative of $F$ at $W(\cdot,\cdot)$. The $L^2$ norm of  $D_{{\rm Mal}}F(W(\cdot,\cdot))$ may be written as
\be \label{G3}
 \|D_{{\rm Mal}}F(W(\cdot,\cdot))\|^2_2  \ = \  \int_0^T \|D_{{\rm Mal},t}F(W(\cdot,\cdot))\|_{2,L}^2 \ dt \ ,
\ee
where $D_{{\rm Mal},t}F(W(\cdot,\cdot))$ is in $\ell_2(Q_L,\mathbb{R})$.
We may obtain from (\ref{F3}) a  Poincar\'{e} inequality for functions of  $\tilde{\om}_L(\cdot,t), \ 0<t<T$. The derivative of a function $F(\tilde{\om}_L (\cdot,t): \ 0<t<T)$ with respect to 
$\tilde{\om}_L$ is defined, as in the case of the Malliavin derivative, by directional derivatives.  Thus for $g\in L^2([0,T]\times Q_L,\mathbb{R}^d)$ one has
\begin{multline} \label{H3}
D_{\tilde{\om}_L,g}F(\tilde{\om}_L) \ = \ \lim_{\eta\ra0} \eta^{-1}\left[   F(\tilde{\om}_L (\cdot,\cdot) +\eta g(\cdot,\cdot) )  -     F(\tilde{\om}_L (\cdot,\cdot)  )               \right] \\
 = \ 
\left[ D_{     \tilde{\om}_L}F(\tilde{\om}_L (\cdot,\cdot)) ,g              \right] \
= \ \int_0^T\left[ D_{     \tilde{\om}_L,t}F(\tilde{\om}_L (\cdot,\cdot)), g(\cdot,t)              \right]_L \ dt \ . 
\end{multline}
We see from (\ref{C3}), (\ref{H3})  that
\be \label{I3}
D_{{\rm Mal}, t} F(\tilde{\om}_L (\cdot,\cdot)) \ = \  (-\De+m^2)^{-1/2}\na^*D_{\tilde{\om}_L,t} F(\tilde{\om}(\cdot,\cdot)) \ , \quad 0<t<T \ .
\ee
Since $(-\De+m^2)^{-1/2}\na^*: \ell_2(Q_L,\mathbb{R}^d)\ra \ell_2(Q_L,\mathbb{R})$ is bounded with norm $\|(-\De+m^2)^{-1/2}\na^*\|_{2,L}\le 1$,   it follows from (\ref{F3}), (\ref{I3}) that  
\begin{multline} \label{J3}
{\rm var} \left[ F(\tilde{\om}_L(\cdot,t): \ 0<t<T)\right] \\
 \le \  \langle  \| (-\De+m^2)^{-1/2}\na^*D_{\tilde{\om}_L}F(\tilde{\om}_L (\cdot,\cdot))\|^2_2          \rangle \  \le \   \langle  \| D_{\tilde{\om}_L}F(\tilde{\om}_L (\cdot,\cdot))\|^2_2 \  \rangle  \ .
\end{multline}
\begin{proposition}
Let  $a(\cdot), h(\cdot),\xi(\cdot)$ be in $\ell_2(Q_L,\mathbb{R}^d)$ and $\om^\xi_{\ve,m,h,L}(\cdot,t) , \ t>0,$ the solution to (\ref{C3}) with initial condition  $\om^\xi_{\ve,m,h,L}(\cdot,t)=\xi$. 
The following inequalities hold:
\be \label{K3}
\ve^{-1}{\rm var}\left\{  [a(\cdot),\om^\xi_{\ve,m,h,L}(\cdot,T)]_L \ \right\} \ \le \  \la^{-1}\|(-\De+m^2)^{-1/2}\na^*a\|_{2,L}^2 \ , \quad T\ge 0 \ ,
\ee
\be \label{X3}
\left|\left\langle \ [a(\cdot), \om^0_{\ve,m,h,L}(\cdot,T)]_L \ \right\rangle  \ \right| \ \le \ \la^{-1} \|(-\De+m^2)^{-1/2}\na^*a\|_{2,L}\|(-\De+m^2)^{-1/2}\na^*h\|_{2,L} \ , \quad T\ge 0 \ ,
\ee
\begin{multline} \label{Z3}
\left\langle \ \exp\left[\frac{1}{\sqrt{\ve}}\left\{[a(\cdot),\om^\xi_{\ve,m,h,L}(\cdot,T)]_L-\langle \ [a(\cdot),\om^\xi_{\ve,m,h,L}(\cdot,T) ]_L \ \rangle \right\} \ \right] \ \right\rangle \\ \le \ 
\exp\left[  \frac{\|(-\De+m^2)^{-1/2}\na^*a\|_{2,L}^2}{2\la}                \right] \ , \quad T\ge 0.
\end{multline}
\end{proposition}
\begin{proof}
To prove (\ref{K3}) we apply the Poincar\'{e} inequality (\ref{J3}).  An equation for the derivatives
$D_{\tilde{\om}_L}\om^\xi_{\ve,m,h,L}(x,t), \ x\in Q_L,t>0, $ may be obtained from first variation analysis of the SDE (\ref{C3}).  Integrating (\ref{C3}) we have that
\begin{multline} \label{L3}
\om^\xi_{\ve,m,h,L}(\cdot,t) \ = \ e^{-t/2}\xi -\left\{1-e^{-t/2}\right\}\na[-\De+m^2]^{-1}\na^*h(\cdot)\\
-\frac{1}{2} \int_0^te^{-(t-s)/2} \ ds \ \na[-\De+m^2]^{-1}\na^*\left[ V'( \om^\xi_{\ve,m,h,L}(\cdot,s)) - \om^\xi_{\ve,m,h,L}(\cdot,s)\right] \\
+\sqrt{\ve}  \int_0^te^{-(t-s)/2}  \tilde{\om}_L(\cdot,s) \ ds   \ .
\end{multline}
Then we have on differentiating (\ref{L3}) in direction $g$  that
\begin{multline} \label{M3}
D_{     \tilde{\om}_L,g}\om^\xi_{\ve,m,h,L}(\cdot,t) \ = \ 
\frac{1}{2} \int_0^te^{-(t-s)/2} \ ds \  \na[-\De+m^2]^{-1}\na^*\mathbf{b}( \om^\xi_{\ve,m,h,L}(\cdot,s))D_{     \tilde{\om}_L,g}\om^\xi_{\ve,m,h,L}(\cdot,s)  \\
+\sqrt{\ve}  \int_0^te^{-(t-s)/2} g(\cdot,s) \ ds  \ .
\end{multline}
The function $\mathbf{b}(\cdot)$ on $\mathbb{R}^d$ with range in the symmetric $d\times d$ matrices is defined by $V''(\cdot)=I_d-\mathbf{b}(\cdot)$, and from (\ref{D1}) satisfies the quadratic form inequality
\be \label{N3}
0 \ \le \mathbf{b}(\cdot) \ \le \ (1-\la)I_d \ .
\ee

We may solve (\ref{M3}) for $D_{     \tilde{\om}_L,g}\om^\xi_{\ve,m,h,L}(\cdot,t), \ 0\le t\le T$,  in the Banach space $\mathcal{E}_T$ of continuous functions $f:Q_L\times [0,T]\ra\mathbb{R}^d$ with norm
$\|f(\cdot,\cdot)\|_{\mathcal{E}_T}=\sup_{0\le t\le T} \|f(\cdot,t\|_{2,L}$.  Let $\mathcal{L}_T$ be a linear operator on $\mathcal{E}_T$  defined by 
\be \label{O3}
\mathcal{L}_Tf(\cdot,t) \ = \ \frac{1}{2} \int_0^te^{-(t-s)/2} \ ds \  \na[-\De+m^2]^{-1}\na^*\mathbf{b}( \om^\xi_{\ve,m,h,L}(\cdot,s))f(\cdot,s) \ , \quad 0\le t\le T \ .
\ee
We see from (\ref{N3}) that $\|\mathcal{L}_Tf(\cdot,\cdot)\|_{\mathcal{E}_T}\le (1-\la)\|f(\cdot,\cdot)\|_{\mathcal{E}_T}$, whence  
$\|\mathcal{L}_T\|_{\mathcal{E}_T}\le (1-\la)$.
Then $f(\cdot,t)=\ve^{-1/2}D_{     \tilde{\om}_L,g}\om^\xi_{\ve,m,h,L}(\cdot,t), \ 0\le t\le T,$ is the solution to the affine fixed point equation
\be \label{P3}
f(\cdot,\cdot) \ = \ \mathcal{A}_Tf(\cdot,\cdot) \ = \  \mathcal{L}_Tf(\cdot,\cdot)+ k(\cdot,\cdot), \quad k(\cdot,t) \ = \  \int_0^te^{-(t-s)/2}  g(\cdot,s) \ ds   \ , \ \ 0\le t\le T.
\ee
By the contraction mapping theorem it follows that for $0\in\mathcal{E}_T$  the null vector,  one has
$\lim_{n\ra\infty} \| \mathcal{A}_T^n0-\ve^{-1/2}D_{     \tilde{\om}_L,g}\om^\xi_{\ve,m,h,L}(\cdot,\cdot)\|_{\mathcal{E}_T}=0$. We conclude that
\be \label{Q3}
\lim_{n\ra\infty}[a(\cdot), \mathcal{A}_T^n0(\cdot,T)]_L \ = \ \ve^{-1/2}[a(\cdot), D_{     \tilde{\om}_L,g}\om^\xi_{\ve,m,h,L}(\cdot,T)]_L  \ .
\ee

We may obtain from (\ref{Q3}) a representation for the  derivative  $ D_{     \tilde{\om}_L,t}F(\tilde{\om}_L(\cdot,\cdot)), \ 0\le t\le T,$ of the function 
\be \label{R3}
 F(\tilde{\om}_L(\cdot,t):  \ 0<t<T) \ = \ \ve^{-1/2}[a(\cdot),\om^\xi_{\ve,m,h,L}(\cdot,T)]_L \ .
\ee
To do this we write
\be \label{S3}
[a(\cdot), \mathcal{A}_T^n0(\cdot,T)]_L \ = \  \int_0^Te^{-(T-t)/2} [a_n(\cdot,t,T),g(\cdot,t)]_L  \ dt \ .
\ee
It is easy to see from (\ref{S3}) that $a_1(\cdot,t,T)=  a(\cdot), \ 0\le t\le T$.  More generally we have for $n=1,2,\dots,$  that
\begin{multline} \label{T3}
a_{n+1}(\cdot,t,T) \ = \ a(\cdot) \\
+\sum_{k=1}^n\int_{t<s_1<s_2<\cdots <s_k<T}ds_1\cdots ds_k \ \prod_{k'=1}^k\left[\frac{1}{2} \mathbf{b}( \om^\xi_{\ve,m,h,L}(\cdot,s_{k'})) \na[-\De+m^2]^{-1}\na^*\right] a(\cdot) \ .
\end{multline}
The $k$th term in (\ref{T3}) has norm in $\ell_2(Q_L,\mathbb{R}^d)$ bounded by $\|a\|_{2,L}[(1-\la)(T-t)/2]^k/k!$.  Hence $a_n(\cdot,t,T)$ converges in $\ell_2(Q_L,\mathbb{R}^d)$ 
 as $n\ra\infty$ to a function   $a_\infty(\cdot,t,T)$ and 
 \be \label{U3}
 \|(-\De+m^2)^{-1/2}\na^*a_\infty(\cdot,t,T)\|_{2,L} \ \le \  \exp[(1-\la)(T-t)/2]\|(-\De+m^2)^{-1/2}\na^*a\|_{2,L} \ .
 \ee
  It also follows from (\ref{R3})-(\ref{T3}) that
 \be \label{V3}
 D_{     \tilde{\om}_L,t}F(\tilde{\om}_L) \ = \ \exp[-(T-t)/2]a_\infty(\cdot,t,T), \ 0\le t\le T \ .
 \ee
  We conclude then from (\ref{J3}), (\ref{U3}), (\ref{V3}) that
 \begin{multline} \label{W3}
 {\rm var} \left[ F(\tilde{\om}_L(\cdot,\cdot))\right] \ \le \ \int_0^T e^{-\la(T-t)} \ dt  \ \|(-\De+m^2)^{-1/2}\na^*a\|^2_{2,L}  \\
  = \ \frac{1-e^{-\la T}}{\la} \|(-\De+m^2)^{-1/2}\na^*a\|^2_{2,L}  \ ,
 \end{multline} 
whence (\ref{K3}) follows. 

To prove (\ref{X3}) we replace $h(\cdot)$ in (\ref{C3}) by $\al h(\cdot)$ with $0\le \al\le 1$ and set $g(\al)=\left\langle \ [a(\cdot), \om^0_{\ve,m,\al h,L}(\cdot,T)]_L \ \right\rangle $.
Since $V(\cdot)$ is even we have that $g(0)=0$.  We take the $\al$ derivative of (\ref{L3}) with $h(\cdot)$ replaced by  $\al h(\cdot)$. Letting $ \mathcal{L}_{\al,T}$ be the operator (\ref{O3}) with    $ \om^\xi_{\ve,m,h,L}(\cdot,t)$ replaced by $ \om^0_{\ve,m,\al h,L}(\cdot,t)$, we have that
\begin{multline} \label{Y3}
f_\al(\cdot,t) \ = \ \mathcal{L}_{\al,T}f_\al(\cdot,t)-\left\{1-e^{-t/2}\right\}\na[-\De+m^2]^{-1}\na^*h(\cdot) \ , \\
{\rm where \ } f_\al(\cdot,t) \ = \ \frac{d}{d\al} \om^0_{\ve,m,\al h,L}(\cdot,t)  \ , \quad 0\le t\le T \ .
\end{multline}
By the contraction mapping theorem there is a unique solution $f_\al$  in $\mathcal{E}_T$  to (\ref{Y3})  and $\|f_\al\|_{\mathcal{E}_T}\le \la^{-1}\|h(\cdot)\|_{2,L}$. Hence $|g'(\al)|=
\left|\left\langle \ [a(\cdot), f_\al(\cdot,T)]_L \ \right\rangle \ \right| \ \le \ \la^{-1}\|a\|_{2,L}\|h(\cdot)\|_{2,L}$. This shows via the mean value theorem that the LHS of (\ref{X3}) is bounded by
$ \la^{-1}\|a\|_{2,L}\|h(\cdot)\|_{2,L}$, which is weaker than (\ref{X3})  The actual  inequality (\ref{X3}) follows in a similar way. 

To prove (\ref{Z3}) we use the Clark-Okone formula \cite{ct,nu}. For $t>0$ let $\mathcal{F}_t$ be the $\sig$-field generated by Brownian motion $B(x,s), \ x\in Q_L, \ 0\le s\le t$.  
Let $F(\tilde{\om}_L)$ be a function of $\tilde{\om}_L(x,t),  \ x\in Q_L, \ 0\le t\le T$.  The Martingale representation theorem  yields a formula
\be \label{AA3}
F(\tilde{\om}_L)-\langle \ F(\tilde{\om}_L)  \ \rangle   \ = \ \int_0^T[\sig_L(\cdot,t,T),dB(\cdot,t)]_L \  ,
\ee
where  $\sig_L(t,T))\in\ell_2(Q_L,\mathbb{R})$ is measurable  with respect to $\mathcal{F}_t, \ 0\le t\le T$. We then have that
\be \label{AB3}
\left\langle \ \exp\left[        \  F(\tilde{\om}_L)-\langle \ F(\tilde{\om}_L)  \ \rangle  \ \right] \ \right\rangle \
= \  \left\langle \ \exp\left[ \frac{1}{2} \int_0^T dt \   \left\|  \sig_L(t,T) \right\|_{2,L}^2   \       \right]  \ \right\rangle \ .
\ee
The Clark-Okone formula tells us that
\be \label{AC3}
\sig_L(t,T) \ = \  E\left[     D_{{\rm Mal},t}F(\tilde{\om}_L) \ \big| \ \mathcal{F}_t                      \right]  \ , \quad 0\le t\le T\  .
\ee
Taking $F(\tilde{\om}_L)$ to be the function (\ref{R3}), it follows from (\ref{I3}), (\ref{V3}) that
\be \label{AD3}
D_{{\rm Mal},t}F(\tilde{\om}_L)  \ = \ \exp[-(T-t)/2] (-\De+m^2)^{-1/2}\na^*a_\infty(\cdot,t,T), \ 0\le t\le T \ .
\ee
The inequality (\ref{Z3}) now follows from (\ref{U3}), (\ref{AB3})-(\ref{AD3}). 
\end{proof}
\begin{proposition}
With the notation in the statement of Proposition 3.1, the following inequality holds:
\be \label{AE3}
\left\|\left\langle \  \om^0_{\ve,m,h,L}(\cdot,T) \ \right\rangle-\left\langle \  \na\phi(\cdot) \ \right\rangle_{\ve,m,h,L} \ \right\|_{2,L} \ \le \ \frac{e^{-\la T/2}}{\la}\sqrt{ \|h\|^2_{2,L}+\ve\la dL^d} \ , \quad T\ge 0  \ .
\ee
\end{proposition}
\begin{proof}
Let $\zeta(\cdot)\in\ell_2(Q_L,\mathbb{R}^d)$ and define the directional derivative $D_\zeta\om^\xi_{\ve,m,h,L}(\cdot,T)$ by
\be \label{AF3}
D_\zeta\om^\xi_{\ve,m,h,L}(\cdot,T) \ = \  \lim_{\eta\ra0} \eta^{-1} \left[    \om^{\xi+\eta \zeta}_{\ve,m,h,L}(\cdot,T) -\om^\xi_{\ve,m,h,L}(\cdot,T)                              \right] \ .
\ee
We see  that $f(\cdot,t)=D_\zeta\om^\xi_{\ve,m,h,L}(\cdot,t), \ 0\le t\le T,$ is the solution to an equation similar to (\ref{P3}), 
\be \label{AG3}
f(\cdot,\cdot) \ = \ \mathcal{A}_Tf(\cdot,\cdot) \ = \  \mathcal{L}_Tf(\cdot,\cdot)+ k(\cdot,\cdot), \quad k(\cdot,t) \ = \  e^{-t/2}\zeta(\cdot)  \ , \ \ 0\le t\le T.
\ee
Since $\mathcal{A}_T$ is a contraction on $\mathcal{E}_T$ we have similarly to (\ref{Q3}) that
\be \label{AH3}
\lim_{n\ra\infty}[a(\cdot), \mathcal{A}_T^n0(\cdot,T)]_L \ = \ [a(\cdot), D_\zeta\om^\xi_{\ve,m,h,L}(\cdot,T)]_L  \ .
\ee
Next we observe that
\be \label{AI3}
[a(\cdot), \mathcal{A}_T^n0(\cdot,T)]_L \ = \ e^{-T/2}[a_n(\cdot,0,T), \zeta(\cdot)]_L \ , \quad n=1,2,\dots, \  T>0 \ ,
\ee
where $a_n(\cdot,t,T)$ is defined by (\ref{T3}).  We conclude from (\ref{AH3}), (\ref{AI3}) that
\be \label{AJ3}
 [a(\cdot), D_\zeta\om^\xi_{\ve,m,h,L}(\cdot,T)]_L  \ = \ e^{-T/2}[a_\infty(\cdot,0,T),\zeta(\cdot)]_L \ .
\ee
Similarly to  (\ref{U3}) we see that $\|a_\infty(0,T)\|_{2,L}\le e^{(1-\la)T/2}\|a\|_{2,L}$. It follows from this and  (\ref{AJ3}) that
\be \label{AK3}
\|D_\zeta\om^\xi_{\ve,m,h,L}(\cdot,T)\|_{2,L} \ \le \ e^{-\la T/2}\|\zeta\|_{2,L} \ .
\ee
Using the identity
\be \label{AL3}
\om^\xi_{\ve,m,h,L}(\cdot,T)-\om^0_{\ve,m,h,L}(\cdot,T) \ = \ \int_0^1 D_\xi\om^{\al\xi}_{\ve,m,h,L}(\cdot,T)  \ d\al \ ,
\ee
we conclude from (\ref{AK3}) that
\be \label{AM3}
\|\om^\xi_{\ve,m,h,L}(\cdot,T)-\om^0_{\ve,m,h,L}(\cdot,T)  \|_{2,L} \ \le \  e^{-\la T/2}\|\xi\|_{2,L} \ .
\ee

To prove (\ref{AE3}) we first observe from (\ref{O2}), (\ref{X3}) that
\be \label{AN3}
\left\|\langle \  \na\phi(\cdot) \ \rangle_{\ve,m,h,L}\right\|_{2,L}, \quad \left\|\left\langle \  \om^0_{\ve,m,h,L}(\cdot,T) \ \right\rangle  \ \right\|_{2,L} \ \le \ \la^{-1} \|h\|_{2,L} \ , \quad T>0 \ .
\ee
Next we observe  since the measure $\langle\cdot\rangle_{\ve,m,h,L}$ is invariant for the stochastic dynamics (\ref{B3}) that
\be \label{AO3}
\left\langle \ [a(\cdot), \na\phi(\cdot)]_L \ \right\rangle_{\ve,m,h,L}  \ = \ \left\langle \ [a(\cdot), \om^\xi_{\ve,m,h,L}(\cdot,T)]_L \ \right\rangle_{\ve,m,h,L} \  , \quad T\ge 0 \ ,
\ee
where $\xi(\cdot)=\na\phi(\cdot)\in\ell_2(Q_L,\mathbb{R}^d)$ is a random variable, independent of the dynamics (\ref{B3}),  and the distribution of $\phi(\cdot)$ is determined by the measure
$\langle\cdot\rangle_{\ve,m,h,L}$.
It follows from (\ref{AM3}), (\ref{AO3}) that the square of the LHS of (\ref{AM3}) is bounded by
\begin{multline} \label{AP3}
\left\langle \  \|\om^\xi_{\ve,m,h,L}(\cdot,T)-\om^0_{\ve,m,h,L}(\cdot,T)  \|^2_{2,L} \ \right\rangle_{\ve,m,h,L} \ \le  \ e^{-\la T}\left\langle \  \|\na\phi(\cdot)\|_{2,L}^2 \ \right\rangle_{\ve,m,h,L} 
\\  = \  e^{-\la T}\left\| \ \langle \na\phi(\cdot) \ \rangle_{\ve,m,h,L} \ \right\|_{2,L}^2+e^{-\la T}\sum_{k=1}^\infty{\rm var}_{\ve,m,h,L}\left\{[a_k(\cdot),\na\phi(\cdot)] \right\} \ ,
\end{multline}
where $a_k(\cdot), \ k=1,2,\dots,$ is an orthonormal basis for $\ell_2(Q_L,\mathbb{R}^d)$. The inequality (\ref{AE3}) follows from (\ref{G1}), (\ref{AN3}), (\ref{AP3}). 
\end{proof}
\begin{rem}
Note that the inequality (\ref{AE3}) depends on the dimension $dL^d$ except in the case $\ve=0$.  The  inequality (\ref{AE3}) with $\ve=0$ also holds when $\ve>0$   in the SDE (\ref{B3}) provided  $V''(\cdot)$ is assumed constant  i.e. the Gaussian case since the mean $\left\langle \  \om^\xi_{\ve,m,h,L}(\cdot,T) \ \right\rangle $ evolves with $T$ as in the deterministic case $\ve=0$. It follows from (\ref{AN3}) that in the general case  there is a  bound  uniform in $T$ on the LHS of (\ref{AE3}) which is independent of $L$.  However we are unable to show that  the bound converges to $0$ as $T\ra\infty$. 
\end{rem}

The upper bound (\ref{N2}), (\ref{O2}) on $|D_aq_{\ve,m,L}(h)|$ follows from (\ref{X3})  and Proposition 3.2. The inequality is actually slightly weaker than (\ref{O2}) since $m^2$ in (\ref{O2}) is replaced by $\la m^2\le m^2$  in (\ref{X3}).  However the inequalities become identical in the limit $L\ra\infty, \ m\ra0$.  We may also derive (\ref{G1})  from (\ref{K3}) by  using the identity for two random variables $X,Y$:
\be \label{AQ3}
{\rm Var}[X] \ = \  E\big[ {\rm Var}[X|Y]\big]+ {\rm Var}\big[E[X|Y]\big] \ . 
\ee
Let $\xi(\cdot)\in\ell_2(Q_L,\mathbb{R}^n)$ in (\ref{K3}) be a random variable, independent of the dynamics (\ref{B3}), which satisfies $E[\|\xi\|_{2,L}^2]<\infty$. We set $Y=\xi(\cdot)$ and $X=\ve^{-1/2}[a(\cdot),\om^\xi_{\ve,m,h,L}(\cdot,T)]_L $.
From (\ref{K3}) we have that
\be \label{AR3}
 E\big[ {\rm Var}[X|Y]\big] \ \le \ \la^{-1}\|(-\De+m^2)^{-1/2}\na^*a\|_{2,L}^2 \ .
\ee
We also have that
\be \label{AS3}
E[X|Y] \ = \ E\left[    \ve^{-1/2}[a(\cdot),\om^\xi_{\ve,m,h,L}(\cdot,T)-\om^0_{\ve,m,h,L}(\cdot,T)]_L \ \Big|  \ \xi(\cdot)  \right] + E\left[    \ve^{-1/2}[a(\cdot),\om^0_{\ve,m,h,L}(\cdot,T)]_L  \right]  \ .
\ee
Since the second term on the RHS of (\ref{AS3}) is independent of $\xi$ we have using (\ref{AM3}) that
\be \label{AT3}
{\rm Var}\big[E[X|Y]\big]  \ \le \  Ce^{-\la T} \quad {\rm for \ some \ constant \ } C \ .
\ee
The inequality (\ref{G1}) follows from (\ref{AQ3})-(\ref{AT3})  by choosing $\xi(\cdot)=\na\phi(\cdot)$, with the distribution of $\phi(\cdot)$ determined by the measure
$\langle\cdot\rangle_{\ve,m,h,L}$, and letting $T\ra\infty$.  The exponential inequality (\ref{E3}) can similarly be derived from the exponential inequality (\ref{Z3}) by writing
\be \label{AU3}
\langle \ [a(\cdot),\om^\xi_{\ve,m,h,L}(\cdot,T) ]_L  \ \rangle \ = \ \langle \ [a(\cdot),\om^0_{\ve,m,h,L}(\cdot,T) ]_L \ \rangle 
+\langle \ [a(\cdot),\om^\xi_{\ve,m,h,L}(\cdot,T)-\om^0_{\ve,m,h,L}(\cdot,T) ]_L \ \rangle  \ ,
\ee
using (\ref{AE3}), (\ref{AM3}) and letting $T\ra\infty$. 

The lower bound (\ref{S2}), which implies the upper bound in (\ref{L1}), may also be derived from the properties of solutions to (\ref{B3}).  To see this we consider the general situation (\ref{R1})-(\ref{V1}). 
Let $W:\mathbb{R}^n\ra\mathbb{R}$ be a $C^2$ convex function with Hessian $D^2W(\cdot)$ satisfying the quadratic form  inequalities
\be \label{AV3}
0 \ < \ D^2W(\cdot) \ \le \ A^{-1} \ , \quad {\rm where \ }A \ {\rm is \ symmetric \ positive \ definite.}
\ee
For any $k\in\mathbb{R}^n$ we denote by $W_k(\cdot)$  the function $W_k(\phi)=W(\phi)+[k,\phi]_n, \ \phi\in\mathbb{R}^n$, where $[\cdot,\cdot]_n$ is the Euclidean inner product on $\mathbb{R}^n$. Taking $b(\cdot)$ in (\ref{V1}) to be $b(\phi)=-\frac{1}{2}DW_k(\phi), \ \phi\in\mathbb{R}^n$,  we have already observed that the probability measure with density proportional to the function $\phi\ra\exp[-W_k(\phi)/\ve], \ \phi\in\mathbb{R}^n,$ is invariant for the stochastic dynamics (\ref{V1}). We define directional derivatives of functions $F(k), \ k\in\mathbb{R}^n,$ by
\be \label{AW3}
D_bF(k) \ = \ \lim_{\eta\ra0}\eta^{-1}[F(k+\eta b)-F(k)] \  , \quad b\in\mathbb{R}^n \ .
\ee
Letting $\langle \cdot\rangle_{\ve,k}$ denote expectation with respect to the probability measure with density proportional to $\phi\ra\exp[-W_k(\phi)/\ve], \ \phi\in\mathbb{R}^n,$ we define a linear operator
$\mathcal{L}_{\ve,k,\infty}$ on $\mathbb{R}^n$ by
\be \label{AX3}
\mathcal{L}_{\ve,k,\infty} b \ = \ -D_b \langle \ \phi \ \rangle_{\ve,k} \ = \ \ve^{-1}\left\langle \ \left\{[b,\phi]_n-\langle\  [b,\phi]_n \ \rangle_{\ve,k} \ \right\}\phi \ \right\rangle_{\ve,k} \ , \quad b\in\mathbb{R}^n \ .
\ee
We see from (\ref{AX3}) that $\mathcal{L}_{\ve,k,\infty}$ is self-adjoint and 
\be \label{AY3}
[b_1, \mathcal{L}_{\ve,k,\infty}b_2]_n \ = \ \ve^{-1}{\rm cov}_{\ve,k}\left\{  [b_1,\phi]_n, \ [b_2,\phi]_n                 \right\} \ , \quad b_1,b_2\in\mathbb{R}^n \ .
\ee
It follows from (\ref{AY3}) that  $\mathcal{L}_{\ve,k,\infty}$ is positive definite.

The SDE (\ref{V1}) with $b(\cdot)=-\frac{1}{2}DW_k(\cdot)$ is given by
 \be \label{AZ3}
d\phi_{\ve,k}(t)  \ = \  -\frac{1}{2}A\left[DW(\phi_{\ve,k}(t))+k  \right] \ dt+\sqrt{\ve A} \  dB(t) \ , \quad t>0 \ .
\ee
We consider solutions of (\ref{AZ3}) with initial condition $\phi_{\ve,k}(0) =0$ and
take the directional derivative of (\ref{AZ3}) with respect to $k$ as in (\ref{AW3}). We see that $D_b\phi_{\ve,k}(t), \ t>0,$ is  the solution to the the first variation linear  initial value problem 
\be \label{BA3}
d [D_b\phi_{\ve,k}(t)]  \ = \  -\frac{1}{2}A\left[D^2W(\phi_{\ve,k}(t))D_b\phi_{\ve,k}(t)+b  \right] \ dt \ , \quad D_b\phi_{\ve,k}(0)=0 \ .
\ee
In view of (\ref{AV3}) we may write (\ref{BA3}) as
\be \label{BB3}
d [D_b\phi_{\ve,k}(t)]  \ = \  -\frac{1}{2}\left[D_b\phi_{\ve,k}(t) -AK(\phi_{\ve,k}(t))D_b\phi_{\ve,k}(t) +Ab\right] \ dt \ , \quad D_b\phi_{\ve,k}(0)=0 \ ,
\ee
where $K(\phi)=A^{-1}-D^2W(\phi), \ \phi\in\mathbb{R}^n,$ is a self-adjoint operator on $\mathbb{R}^n$ and satisfies the quadratic form inequalities
\be \label{BC3}
0 \ \le \ K(\cdot)\ \ < \ A^{-1} \ .
\ee
We may integrate (\ref{BB3}) to obtain the integral equation
\be \label{BD3}
D_b\phi_{\ve,k}(t) \ = \ \frac{1}{2}\int_0^t e^{-(t-s)/2}AK(\phi_{\ve,k}(s))D_b\phi_{\ve,k}(s) \ ds -\left[1-e^{-t/2}\right]Ab \ , \quad t>0.
\ee
By iterating the integral equation (\ref{BD3}) we may expand $D_b\phi_{\ve,k}(T)$ in a power series in $K(\cdot)$,
\begin{multline} \label{BE3}
D_b\phi_{\ve,k}(T) \ = \ \sum_{r=0}^\infty D_b\phi_{\ve,k,r}(T)  \ , \quad D_b\phi_{\ve,k,0}(T)= -\left[1-e^{-T/2}\right]Ab  \ , \\
 D_b\phi_{\ve,k,r}(T) \ = \  -\frac{1}{2^{r+1}}\int_{0<s_1<s_2<\cdots<s_{r+1}<T} ds_1\cdots ds_{r+1} \ e^{-(T-s_1)/2}AK(\phi_{\ve,k}(s_{r+1}))\cdots AK(\phi_{\ve,k}(s_2))Ab \ .
\end{multline} 
It follows from (\ref{BC3}) that the Euclidean matrix norm  of $A^{1/2}K(\cdot)A^{1/2}$ is strictly less than $1$, whence the power series (\ref{BE3}) converges. 

One can easily see that
\be \label{BF3}
\lim_{T\ra\infty} \langle \  \phi_{\ve,k}(T) \ \rangle \ = \ \langle \ \phi \ \rangle_{\ve,k} \ , \quad \lim_{T\ra\infty} \langle \ D_b \phi_{\ve,k}(T) \ \rangle \ = \ D_b\langle \ \phi \ \rangle_{\ve,k}  \ ,
\ee
where the function $t\ra D_b\phi_{\ve,k}(t), \ t>0,$ is defined as the solution to (\ref{BD3}).  
The limits (\ref{BF3}) hold not only for initial data $\phi_{\ve,k}(0)=0$, but for general initial data of $\phi_{\ve,k}(t), \ t>0$, in particular when $\phi_{\ve,k}(0)$ has the stationary distribution
$\phi\ra\exp[-W_k(\phi)/\ve], \ \phi\in\mathbb{R}^n$. In that case $t\ra \phi_{\ve,k}(t), \ t\ge 0,$ may be extended to $t\in\mathbb{R}$ and is time translation and  time reversal invariant.  The time reversal invariance follows from the self-adjointness of the infinitesimal generator $\mathcal{A}_k$ of the diffusion (\ref{AZ3}).  We have for $C^2$ functions $u:\mathbb{R}^n\ra\mathbb{R}$  that
\be \label{BG3}
\mathcal{A}_ku(\phi) \ = \ -\frac{1}{2} [ADW_k(\phi), Du(\phi)]_n+\frac{\ve}{2}{\rm Tr}[AD^2u(\phi)] \ , \quad \phi\in\mathbb{R}^n \ .
\ee
The self-adjointness and negative definiteness of the operator (\ref{BG3}) follows from the identity
\be \label{BH3}
[v,\mathcal{A}_ku]_{\ve,k} \ = \ \langle \ v(\cdot),\mathcal{A}_ku(\cdot) \ \rangle_{\ve,k} \ = \  -\frac{\ve}{2}\langle \ [Dv(\cdot), ADu(\cdot)]_n \ \rangle_{\ve,k} \ ,
\ee
for $C^2$ functions $u,v:\mathbb{R}^n\ra\mathbb{R}$. 

We define  linear operators $\mathcal{L}_{\ve,k,T}, \ T>0,$ on $\mathbb{R}^n$ by
\be \label{BI3}
\mathcal{L}_{\ve,k,T} b \ = \  -\langle \ D_b\phi_{\ve,k}(T) \ \rangle \  , \quad b\in\mathbb{R}^n \ ,
\ee
where $t\ra D_b\phi_{\ve,k}(t)$ is the solution to (\ref{BD3}) and $t\ra\phi_{\ve,k}(t), \ t\in\mathbb{R},$ is the stationary process for (\ref{AZ3}) with invariant measure $\phi\ra\exp[-W_k(\phi)/\ve], \ \phi\in\mathbb{R}^n$. We have from (\ref{BE3}) that $\mathcal{L}_{\ve,k,T}$ may be written in a power series expansion in $K(\cdot)$,
\be \label{BJ3}
\mathcal{L}_{\ve,k,T} \ = \ \sum_{r=0}^\infty \mathcal{L}_{\ve,k,T,r} \ ,  \quad  \mathcal{L}_{\ve,k,T,r}b  \ = \  -\langle \ D_b\phi_{\ve,k,r}(T) \ \rangle \ , \quad b\in\mathbb{R}^n, \ r=0,1,\dots
\ee
Evidently $ \mathcal{L}_{\ve,k,T,0}=[1-e^{-T/2}]A$ is self-adjoint positive definite.  The operators $ \mathcal{L}_{\ve,k,T,r}, \ r=1,2,\dots,$ are also self-adjoint.  We can see this from (\ref{BE3}) by using the time translation and time reversal invariance of $t\ra\phi_{\ve,k}(t), \ t\in\mathbb{R}$.  Thus we have from (\ref{BE3}) that
\begin{multline} \label{BK3}
 \mathcal{L}_{\ve,k,T,r} \ = \\
   \frac{1}{2^{r+1}}\int_{0<s_1<s_2<\cdots<s_{r+1}<T} ds_1\cdots ds_{r+1} \ e^{-(T-s_1)/2}\left\langle \ AK(\phi_{\ve,k}(T+s_1-s_{r+1}))\cdots AK(\phi_{\ve,k}(T+s_1-s_2))A \ \right\rangle  \\
   =  \ \frac{1}{2^{r+1}}\int_{0<s_1<s_2<\cdots<s_{r+1}<T} ds_1\cdots ds_{r+1} \ e^{-(T-s_1)/2}\left\langle \ AK(\phi_{\ve,k}(s_2))\cdots AK(\phi_{\ve,k}(s_{r+1}))A  \ \right\rangle \\
   = \  \mathcal{L}^*_{\ve,k,T,r} \ , \quad {\rm the \ adjoint \  of  \ }  \mathcal{L}_{\ve,k,T,r} \ .
\end{multline}

It appears that the operators $ \mathcal{L}_{\ve,k,T,r}, \  r\ge 3,$ are not positive definite in general for finite $T$. However the operators $ \mathcal{L}_{\ve,k,\infty,r}= \lim_{T\ra\infty}\mathcal{L}_{\ve,k,T,r}$ are positive definite for all $r=0,1,\dots$.  To see this we consider the function $u:\mathbb{R}^n\times\mathbb{R}^+\ra\mathbb{R}^n$ defined by
\be \label{BL3}
u(\phi,t) \ = \ -E[D_b\phi_{\ve,k}(t) \ | \ \phi(0)=\phi] \ , \quad \phi\in\mathbb{R}^n, \ t>0 \ ,
\ee
where $\phi_{\ve,k}(t), \ t>0,$ is the solution to the SDE (\ref{AZ3}) and $D_b\phi_{\ve,k}(t), \ t>0,$ is the solution to (\ref{BB3}). 
Then $u$ is the solution to the initial value problem
\be \label{BM3}
\frac{\pa u(\phi,t)}{\pa t} \ = \ \mathcal{A}_ku(\phi,t)-\frac{1}{2}[I_n-AK(\phi)]u(\phi,t)+\frac{1}{2}Ab \  , \ t>0, \quad u(\cdot,0)\equiv 0 \ .
\ee
We denote by $K$ the operator on functions $u:\mathbb{R}^n\ra\mathbb{R}^n$ defined by
\be \label{BN3}
Ku(\phi) \ = \ K(\phi)u(\phi) \ , \quad \phi\in\mathbb{R}^n \ .
\ee
We conclude then from (\ref{BI3}), (\ref{BL3})-(\ref{BN3}) that
\be \label{BO3}
\mathcal{L}_{\ve,k,T} b \ = \  \frac{1}{2}  \int_0^T dt \  \left\langle \  \exp\left[ -\frac{1}{2}\{-2\mathcal{A}_k+I_n-AK\}t\right]Ab \ \right\rangle_{\ve,k} \ .
\ee
Letting $T\ra\infty$ in (\ref{BO3}) we obtain the formula
\be \label{BP3}
\mathcal{L}_{\ve,k,\infty} b \ = \   \left\langle \  \{-2\mathcal{A}_k+I_n-AK\}^{-1}Ab \ \right\rangle_{\ve,k} \ .
\ee
Expanding (\ref{BP3}) out in powers of $K$ we see that
\be \label{BQ3}
\mathcal{L}_{\ve,k,\infty,r} \ = \  \left\langle \ \{ (-2\mathcal{A}_k+I_n)^{-1}AK\}^r(-2\mathcal{A}_k+I_n)^{-1}A \ \right\rangle_{\ve,k} \ , \quad r=0,1,\dots
\ee
To see that  $\mathcal{L}_{\ve,k,\infty,r} $ is positive definite we note that
\be \label{BR3}
\left[ b, \mathcal{L}_{\ve,k,\infty,r}b\right]_n \ = \ \left[b, \{(-2\mathcal{A}_k+I_n)^{-1}AK\}^r(-2\mathcal{A}_k+I_n)^{-1}A b\right]_{\ve,k} \ ,
\ee
where $[\cdot,\cdot]_{\ve,k}$ is the inner product  on the Hilbert space of functions  $u:\mathbb{R}^n\ra\mathbb{R}^n$ implied by (\ref{BH3}).  The operators $-\mathcal{A}_k,A,K$ are all self-adjoint  positive-definite on  this Hilbert space and $A$ commutes with $\mathcal{A}_k$. Using the fact that $\mathcal{A}_k$ annihilates the constant function, we have from (\ref{BR3}) that
\be \label{BS3}
\begin{array}{lcl}
\left[ b, \mathcal{L}_{\ve,k,\infty,r}b\right]_n \ &=& \  \left\|                     A^{1/2}(-2\mathcal{A}_k+I_n)^{-1/2}KA\{(-2\mathcal{A}_k+I_n)^{-1}KA\}^{m-1} b                \right\|_{\ve,k}^2 \ , \quad r=2m \ , \\
\left[ b, \mathcal{L}_{\ve,k,\infty,r}b\right]_n \ &=& \  \left\|                     K^{1/2}A\{(-2\mathcal{A}_k+I_n)^{-1}KA\}^m b                \right\|_{\ve,k}^2 \ , \quad r=2m+1 \ ,
\end{array}
\ee
whence we conclude that $\mathcal{L}_{\ve,k,\infty,r}$ is positive definite.

Letting $T\ra\infty$ in (\ref{BJ3}), we obtain from (\ref{BS3})  the quadratic form inequality
\be \label{BT3}
\mathcal{L}_{\ve,k,\infty} \ = \ \sum_{r=0}^\infty \mathcal{L}_{\ve,k,\infty,r} \ \ge \   \mathcal{L}_{\ve,0,\infty,0} \ = \ A \ .
\ee
It follows then from (\ref{AY3}), (\ref{BT3}) that
\be \label{BU3}
\ve^{-1}{\rm var}_{\ve,k}\{ \ [b,\phi]_n \ \} \ \ge \ [b,Ab]_n \ , \quad b\in\mathbb{R}^n \ .
\ee
The inequality (\ref{S2}) follows from (\ref{BU3}) applied to the SDE (\ref{B3}), taking $A=[-\De+m^2]^{-1}$ and $k=\na^*h, \ b=\na^* a$. 

Next we assume the function $V(\cdot)$ is holomorphic in a strip parallel to the real axis  and satisfies the inequality (\ref{O1}).  In that case the SDE (\ref{C3})  may be solved globally in time for some {\it complex} valued $h$. To see this we assume now that $h\in\ell_2(Q_L,\mathbb{C}^d)$ and let  $t\ra \om^\xi_{\ve,m,h,L}(\cdot,t)$ be the solution to (\ref{C3}) with initial condition
$\om^\xi_{\ve,m,h,L}(\cdot,0)=\xi\in \ell_2(Q_L,\mathbb{C}^d)$. The imaginary part  $t\ra \Im  \om^\xi_{\ve,m,h,L}(\cdot,t)$ is evidently  a solution to the initial value problem
\begin{multline} \label{BV3}
d\Im\om^\xi_{\ve,m,h,L}(\cdot,t) \ = \ -\frac{1}{2}\Big[\Im\om^\xi_{\ve,m,h,L}(\cdot,t)+\na[-\De+m^2]^{-1}\na^* \Im h(\cdot) \\
+ \na[-\De+m^2]^{-1}\na^*\{\Im V'(\om_{\ve,m,h,L}(\cdot,t))- \Im\om_{\ve,m,h,L}(\cdot,t)\} \Big] \ dt \ , \quad \Im\om^\xi_{\ve,m,h,L}(\cdot,0) \ = \ \Im\xi \ . 
\end{multline}
By using (\ref{H2}) we may show that the nonlinear term in the evolution equation (\ref{BV3}) is a small perturbation of the linear term provided $\eta>0$ in (\ref{O1}) is sufficiently small. 
\begin{lem} Assume the function $V(\cdot)$ is holomorphic in a strip parallel to the real axis and satisfies the inequalities (\ref{D1}),  (\ref{O1}).  Then
for all   $\eta$ satisfying $0<\eta<\la$ and  $h,\xi\in \ell_2(Q_L,\mathbb{C}^d)$  with $\|\Im h\|_{2,L}, \ \|\Im \xi\|_{2,L}<(\la-\eta)\del(\eta)$, the SDE (\ref{C3}) with initial condition $\xi$  has a  unique strong solution $t\ra \om^\xi_{\ve,m,h,L}(\cdot,t), \ t>0,$  globally in time. Furthermore, the inequality
$\sup_{t>0}\|\Im \om^\xi_{\ve,m,h,L}(\cdot,t)\|_{2,L}< \del(\eta)$ holds.
\end{lem}
\begin{proof}
We follow the argument of \cite{ks}, Chapter 5 to show existence and uniqueness of strong solutions to (\ref{C3}).  Thus we write (\ref{C3}) as the fixed point equation (\ref{L3}). For $T,\del>0$ and 
$\xi\in\ell_2(Q_L,\mathbb{C}^d)$ let   $\mathcal{S}_{\xi,\del,T}$ be the space of continuous functions $f:Q\times[0,T]\ra\mathbb{C}^d$ such that $f(\cdot,0)=\xi(\cdot), \  \sup_{0<t<T}\| f(\cdot,t)\|_{2,L}<\infty$ and $\sup_{0<t<T}\| \Im f(\cdot,t)\|_{2,L}\le \del$. We define a nonlinear operator $\mathcal{K}_T$ on  $\mathcal{S}_{\xi,\del,T}$ by  
\begin{multline} \label{BW3}
\mathcal{K}_Tf(\cdot,t) \ = \ e^{-t/2}\xi -\left\{1-e^{-t/2}\right\}\na[-\De+m^2]^{-1}\na^*h(\cdot)\ \\
-\frac{1}{2} \int_0^te^{-(t-s)/2} \ ds \  \na[-\De+m^2]^{-1}\na^*\left\{V'(f(\cdot,s))-f(\cdot,s)\right\} \ , \quad 0\le t\le T \ .
\end{multline}
Then (\ref{L3}) is given by 
\be \label{BX3}
\om^\xi_{\ve,m,h,L}(\cdot,\cdot) \ = \ \mathcal{K}_T\om^\xi_{\ve,m,h,L}(\cdot,\cdot)+ \sqrt{\ve}k(\cdot,\cdot) \ , \quad  k(\cdot,t) \ = \ \int_0^te^{-(t-s)/2}  \tilde{\om}_L(\cdot,s) \ ds   \ .
\ee
It is clear from (\ref{BW3})  that if $f\in \mathcal{S}_{\xi,\del,T}$ then  $\mathcal{K}_Tf(\cdot,\cdot)$  is continuous, $\mathcal{K}_Tf(\cdot,0)=\xi$ and  $\sup_{0<t<T}\|\mathcal{K}_T f(\cdot,t)\|_{2,L}<\infty$.
 Noting that $|\Im f(x,t)|\le \|\Im f(\cdot,t)\|_{2,L}, \ x\in Q_L, \ 0\le t\le T,$  we further see using (\ref{D1}), (\ref{H2})   that
\begin{multline} \label{BY3}
\|\Im \mathcal{K}_Tf(\cdot,t)\|_{2,L} \ \le \ e^{-t/2}\|\Im\xi\|_{2,L} +\left\{1-e^{-t/2}\right\}\|\Im h\|_{2,L} \\
+ \ \frac{1-\la+\eta}{2} \int_0^te^{-(t-s)/2}\|\Im f(\cdot,s\|_{2,L} \ ds  \ , \quad 0<t<T \ ,
\end{multline}
provided $\del<\del(\eta)$.
We conclude from (\ref{BY3}) that $ \mathcal{K}_T$ maps $\mathcal{S}_{\xi,\del,T}$ to itself if $\eta<\la$ and $\del$ is sufficiently close to $\del(\eta)$. Since the function $k(\cdot,\cdot)$ in (\ref{BX3}) is continuous and real valued,  we may proceed as in \cite{ks}, Chapter 5 to establish  the existence of a  unique strong solution   globally in time to the SDE (\ref{C3}). 
\end{proof}
Next we extend the results of Proposition 3.1 to complex $h(\cdot), \ \xi(\cdot)$ and $a(\cdot)$. 
\begin{proposition}
Assume the function $V(\cdot)$ and  $h(\cdot),\xi(\cdot)\in\ell_2(Q_L,\mathbb{C}^d)$ satisfy the conditions of Lemma 3.1, and $\om^\xi_{\ve,m,h,L}(\cdot,t) , \ t>0,$  is the solution to (\ref{C3}) with initial condition  $\om^\xi_{\ve,m,h,L}(\cdot,t)=\xi$.  Then  for all $a\in\ell_2(Q_L,\mathbb{C}^d)$ the following inequalities hold:
\be \label{BZ3}
\ve^{-1}{\rm var}\left\{  [a(\cdot),\om^\xi_{\ve,m,h,L}(\cdot,T)]_L \ \right\} \ \le \  \frac{1}{\la-\eta}\|(-\De+m^2)^{-1/2}\na^*a\|_{2,L}^2 \ , \quad T\ge 0 \ , 
\ee
\be \label{CA3}
\left|\left\langle \ [a(\cdot), \om^0_{\ve,m,h,L}(\cdot,T)]_L \ \right\rangle  \ \right| \ \le \  \frac{1}{\la-\eta}\|(-\De+m^2)^{-1/2}\na^*a\|_{2,L}\|(-\De+m^2)^{-1/2}\na^*h\|_{2,L} \ , \quad T\ge 0 \ ,
\ee
\begin{multline} \label{CB3}
\left\langle \ \exp\left[\frac{1}{\sqrt{\ve}}\Re\left\{[a(\cdot),\om^\xi_{\ve,m,h,L}(\cdot,T)]_L-\langle \ [a(\cdot),\om^\xi_{\ve,m,,h,L}(\cdot,T) ]_L \ \rangle \right\} \ \right] \ \right\rangle \\
 \le \  \exp\left[  \frac{\|(-\De+m^2)^{-1/2}\na^*a\|_{2,L}^2}{2(\la-\eta)}                \right] \ , \ T\ge 0.
\end{multline}
\end{proposition}
\begin{proof}
We proceed as in the proof of Proposition 3.1, merely extending to the complex case.  The linear operator $\mathcal{L}_T$ defined by  (\ref{O3}) now acts on the complex Banach space $\mathcal{E}_T$ of
continuous functions $f:Q_L\times[0,T]\ra\mathbb{C}^d$.   It follows from Lemma 3.1 and (\ref{O1}) that  $\mathcal{L}_T$  is a bounded operator on  $\mathcal{E}_T$  with norm  $\|\mathcal{L}_T\|_{\mathcal{E}_T}\le (1-\la+\eta)$. The proof of (\ref{BZ3}) follows now as in the proof of (\ref{K3}) of Proposition 3.1. The proofs of (\ref{CA3}), (\ref{CB3})  also follow in a similar way to the proofs of (\ref{X3}), (\ref{Z3}). 
\end{proof}
To extend Proposition 3.2 to complex $h(\cdot)$ we first note that the function $h(\cdot)\ra\left\langle \  \na\phi(\cdot) \ \right\rangle_{\ve,h,L}$ which occurs in (\ref{AE3})   may be analytically continued from functions $h\in\ell_2(Q_L,\mathbb{R}^d)$ to functions $h\in\ell_2(Q_L,\mathbb{C}^d)$ satisfying $\|\Im h(\cdot)\|_{2,L}\le \del_1$ for {\it some} $\del_1>0$. The main point of the next proposition is that 
$\del_1$ may be chosen {\it independent} of $L,m$ as $L\ra\infty, \ m\ra0$. 
\begin{proposition}
Assume the function $V(\cdot)$ satisfies the conditions of Lemma 3.1, and let $g_{\ve,m,L}:\ell_2(Q_L,\mathbb{R}^d)\ra \ell_2(Q_L,\mathbb{R}^d)$ be defined by $g_{\ve,m,L}(h(\cdot))= \left\langle \  \na\phi(\cdot) \ \right\rangle_{\ve,m,h,L}$. Then $g_{\ve,m,L}$ extends analytically  to the strip $\{h\in \ell_2(Q_L,\mathbb{C}^d): \ \|\Im h(\cdot)\|_{2,L}<(\la-\eta)\del(\eta)\}$ with the holomorphic function  
$g_{\ve,m,L}(\cdot)$ taking values in  $\ell_2(Q_L,\mathbb{C}^d)$.  Furthermore, one has that
\be \label{CC3}
\lim_{T\ra\infty}\left\|\left\langle \  \om^0_{\ve,m,h,L}(\cdot,T) \ \right\rangle-g_{\ve,m,L}(h(\cdot)) \ \right\|_{2,L} \ =  0  , \ 
\ee
with uniform convergence in any region $\{h\in \ell_2(Q_L,\mathbb{C}^d): \ \|\Re h(\cdot)\|_{2,L}\le M, \ \|\Im h(\cdot)\|_{2,L}\le (\la-\eta)\del\}, \ M>0, \ 0<\del<\del(\eta)$.
\end{proposition}
\begin{proof}
We choose $\eta,\del$ as in the proof of Lemma 3.1 and show there  is an invariant measure for the solution to (\ref{C3}) when $\|\Im h(\cdot)\|_{2,L}\le (\la-\eta)\del$.  The measure may be constructed in the standard way (see Theorem 4.6.1 of \cite{bs}). Thus for $F: \ell_2(Q_L,\mathbb{C}^d)\ra\mathbb{R}$ a continuous function of compact support we consider the sequence $a_n(F), \ n=1,2,\dots,$ defined by
\be \label{CD3} 
 a_N(F) \ = \ \frac{1}{N}\int_0^N dt \ E\left[     F\left(      \om^0_{\ve,m,h,L}(\cdot,t)          \right)                     \right] \ , \quad N=1,2,\dots
\ee
By using the fact that the Banach space $C_0\left(   \ell_2(Q_L,\mathbb{C}^d)\right)$  of continuous functions on $\ell_2(Q_L,\mathbb{C}^d)$ which decay to $0$ at $\infty$ is separable, and the Riesz representation theorem, we see  there exists a positive Borel measure  $\mu$  on $\ell_2(Q_L,\mathbb{C}^d)$  and  a sequence $N_k, \ k=1,2,\dots,$ such that
\be \label{CE3}
\lim_{k\ra\infty}  a_{N_k}(F) \ = \ \int_{\ell_2(Q_L,\mathbb{C}^d) } F(\cdot) \ d\mu(\cdot) \ , \quad F\in C_c\left(   \ell_2(Q_L,\mathbb{C}^d)         \right) \ .
\ee

To show that $\mu$ is a probability measure we  observe from (\ref{L3}) that
\begin{multline} \label{CF3}
\| \om^0_{\ve,m,h,L}(\cdot,t)   \|_{2,L} \ \le \ \|h(\cdot)\|_{2,L}+\frac{(1-\la+\eta)}{2}\int_0^t e^{-(t-s)/2} \| \om^0_{\ve,m,h,L}(\cdot,s)   \|_{2,L}  \ ds \\
+\sqrt{\ve}\|X(t)\|_{2,L},    \quad {\rm where \ } X(t)=  \int_0^te^{-(t-s)/2}  \tilde{\om}_L(\cdot,s) \ ds   \ .
\end{multline}
Integrating (\ref{CF3}) over the interval $0<t<T$, we obtain the inequality
\begin{multline} \label{CG3}
\frac{1}{T}\int_0^Tdt \ \|\om^0_{\ve,m,h,L}(\cdot,t)   \|_{2,L} \ \le \ \|h(\cdot)\|_{2,L}+\frac{\sqrt{\ve}}{T}\int_0^T dt \ \|X(t)\|_{2,L} \\
+\frac{(1-\la+\eta)}{T}\int_0^T ds \  \| \om^0_{\ve,m,h,L}(\cdot,s)   \|_{2,L}  \  ,
\end{multline}
whence we conclude that
\be \label{CH3}
\frac{1}{T}\int_0^Tdt \ \|\om^0_{\ve,m,h,L}(\cdot,t)   \|_{2,L} \ \le \ \frac{1}{\la-\eta}\left[        \|h(\cdot)\|_{2,L}+\frac{\sqrt{\ve}}{T}\int_0^T dt \ \|X(t)\|_{2,L}                          \right] \ , \quad T>0 \ .
\ee
To estimate the RHS of (\ref{CH3}) we introduce the stationary process $Y(t), \ t\in\mathbb{R},$ defined by
\be \label{CI3}
Y(t) \ = \  \int_{-\infty}^te^{-(t-s)/2}  \tilde{\om}_L(\cdot,s) \ ds   \ , \quad t\in\mathbb{R} \ ,
\ee
where $ \tilde{\om}_L(\cdot,s), \ s<0$, is defined as in (\ref{C3}) with the white noise process $W(\cdot,t), \ t>0,$ being extended to $t\in\mathbb{R}$.   We then have that
\be \label{CJ3}
E\left[         \|X(t)-Y(t)\|^2_{2,L}          \right] \ \le \  L^d e^{-t} \ , \quad t>0 \ .
\ee
It follows from (\ref{CJ3}) that
\be \label{CK3}
E\left[ \left(\frac{1}{T}\int_0^T dt \ \|X(t)-Y(t)\|_{2,L}   \right)^2 \ \right] \ \le \     \frac{4L^d}{T^2} \ .
\ee

We may use (\ref{CK3}) and the maximal ergodic theorem \cite{parry} to estimate the probability that the LHS of (\ref{CH3}) is large. Since the process $t\ra Y(t)$ is stationary the maximal ergodic theorem implies that 
\be \label{CL3}
P\left(      \sup_{N\ge 1} \frac{1}{N}\int_0^N \|Y(t)\|_{2,L}  \ dt > \al                       \right) \ \le \ \frac{1}{\al}E\left[       \int_0^1 \|Y(t)\|_{2,L}  \ dt                  \right] \ , \quad \al>0 \ .
\ee
We can easily estimate the RHS of (\ref{CL3}) by using the Schwarz inequality, whence we have that
\be \label{CM3}
E\left[       \int_0^1 \|Y(t)\|_{2,L}  \ dt                  \right]  \ \le \  \left(\int_0^1 E[\|Y(t)\|_{2,L}^2] \ dt\right)^{1/2} \ \le \  L^{d/2} \ .
\ee
From (\ref{CK3}) and the Chebyshev inequality we have that
\be \label{CN3}
P\left(      \sup_{N\ge 1} \frac{1}{N}\int_0^N \|X(t)-Y(t)\|_{2,L}  \ dt > \al                       \right) \ \le \  \frac{2\pi^2 L^d}{3\al^2} \ .
\ee
We conclude from (\ref{CH3}), (\ref{CL3})-(\ref{CN3}) that
\be \label{CO3}
P\left(      \sup_{N\ge 1} \frac{1}{N}\int_0^N \|\om^0_{\ve,m,h,L}(\cdot,t)   \|_{2,L}   \ dt > \frac{\|h\|_{2,L}+2\sqrt{\ve}\al}{\la-\eta}                       \right) \ \le \  \frac{L^{d/2}}{\al}+ \frac{2\pi^2 L^d}{3\al^2}  \ .
\ee
It follows from (\ref{CO3}) that for any $\ga>1$,
\be \label{CP3}
\sup_{N\ge 1}\frac{1}{N}m\left\{ t\in[0,N]: \      \|\om^0_{\ve,m,h,L}(\cdot,t)   \|_{2,L}   > \frac{\ga[\|h\|_{2,L}+2\sqrt{\ve}\al]}{\la-\eta}                           \right\} \ \ge \frac{1}{\ga} 
\ee
with probability less than the RHS of (\ref{CO3}). We conclude from (\ref{CE3}), (\ref{CP3})  that $\mu$ is a probability measure. 

It is easy to see from the definition (\ref{CE3}) that $\mu$ is an invariant measure for the stochastic dynamics (\ref{C3}). Let $\xi\in \ell_2(Q_L,\mathbb{C}^d)$ be the random variable with distribution given by  $\mu$. Then we need to show that
\be \label{CQ3}
E[F(\xi)] \ = \ E\left[     F\left(      \om^\xi_{\ve,m,h,L}(\cdot,\tau)          \right)                     \right] , \  \tau>0,  \quad F\in C_c\left(   \ell_2(Q_L,\mathbb{C}^d)         \right) \ ,
\ee
assuming $\xi$ and the dynamics $\tau\ra  \om^\xi_{\ve,m,h,L}(\cdot,\tau), \ \tau>0, $  are independent.  This may be proved by conditioning on the Brownian motion (BM) which drives the diffusion in (\ref{CQ3}), which we may  take to be independent of the BM driving the diffusion in (\ref{CD3}). From (\ref{CE3}) we have that
\begin{multline} \label{CR3}
\lim_{k\ra\infty}\frac{1}{N_k}\int_0^{N_k} dt \ E\left[     F\left(      \om^0_{\ve,m,h,L}(\cdot,t+\tau)          \right)          \bigg| \ \tilde{\om}(\cdot,s),  \ 0<s<\tau             \right] \\
= \  E\left[     F\left(      \om^\xi_{\ve,m,h,L}(\cdot,\tau)          \right)    \ \bigg| \ \tilde{\om}(\cdot,s),  \ 0<s<\tau                  \right]  \ , \quad F\in C_c\left(   \ell_2(Q_L,\mathbb{C}^d)         \right) \ .
\end{multline}
 On taking the expectation of the LHS of (\ref{CR3}) with respect to  $\tilde{\om}(\cdot,s),  \ 0<s<\tau$,  we see from  the dominated convergence theorem and (\ref{CE3}) that the limit as $k\ra\infty$ is the LHS of (\ref{CQ3}).  Since the expectation of the RHS of (\ref{CR3}) is equal to the RHS of (\ref{CQ3}), the identity (\ref{CQ3}) follows.
 
The proof of (\ref{CC3}) now follows along the same lines as the proof of Proposition 3.2. We first observe from the definition (\ref{CE3}) and Proposition 3.3 that the invariant variable $\xi$ with distribution measure $\mu$  has finite second moment and satisfies the inequalities
\be \label{CS3}
\| \ \langle \  \xi(\cdot) \ \rangle_{\ve,m,h,L}\|_{2,L} \ \le \ \frac{\|h\|_{2,L}}{\la-\eta} \ , \quad \ve^{-1}{\rm var}_{\ve,h,L}\{[a(\cdot),\xi(\cdot)]_L\} \ \le \  \frac{\|a\|_{2,L}^2}{\la-\eta} \ ,
\ee
for all $a(\cdot)\in\ell_2(Q_L,\mathbb{C}^d)$.  Proceeding as in the proof of Proposition 3.2 we then obtain an inequality similar to (\ref{AE3}), 
\be \label{CT3}
\left\|\left\langle \  \om^0_{\ve,m,h,L}(\cdot,T) \ \right\rangle-\left\langle \  \xi(\cdot) \ \right\rangle_{\ve,m,h,L} \ \right\|_{2,L} \ \le \ \frac{e^{-(\la-\eta) T/2}}{\la-\eta}\sqrt{ \|h\|^2_{2,L}+\ve(\la-\eta) dL^d} \ , \quad T\ge 0  \ .
\ee
It is easy to see from the construction of the function $h(\cdot)\ra \left\langle \  \om^0_{\ve,m,h,L}(\cdot,T) \ \right\rangle$ in the proof of Lemma 3.1 that it is holomorphic in $h(\cdot)$. We conclude then from (\ref{CT3}) that the function $h(\cdot)\ra \left\langle \  \xi(\cdot) \ \right\rangle_{\ve,m,h,L}=g_{\ve,m,L}(h(\cdot))$ is also holomorphic and (\ref{CC3}) holds. 
\end{proof}
\begin{proof}[Proof of Theorem 1.2] As in (\ref{A2}) we apply the FTC to (\ref{H1}). Using (\ref{N2})  we  obtain the formula
\begin{multline} \label{CU3}
q_{\ve,m,L}(h(\cdot))- q_{\ve,m,L}(0) \ = \ \int_0^1 d\al \ \frac{d}{d\al}q_{\ve,m,L}(\al h(\cdot)) \\
 = \  \int_0^1 d\al  \ D_h q_{\ve,m,L}(\al h(\cdot)) \ = \ \int_0^1 d\al  \ \left\langle \      [h(\cdot),\na\phi(\cdot)]_L      \          \right\rangle _{\ve,m,\al h,L}  \ = \ 
 \int_0^1 d\al  \ [h(\cdot),g_{\ve,m,L}(\al h(\cdot))]_L  \ ,
\end{multline}
where $g_{\ve,m,L}(\cdot)$ is defined in the statement of Proposition 3.4.  We conclude from Proposition 3.4 that  the function  $h(\cdot)\ra q_{\ve,m,L}(h(\cdot))$  defined for real $h\in\ell_2(Q_L,\mathbb{R}^d)$ extends analytically to complex $h\in \ell_2(Q_L,\mathbb{C}^d)$ satisfying $\|\Im h(\cdot)\|_{2,L}<(\la-\eta)\del(\eta)$.  

To prove the bounds (\ref{W1}), (\ref{X1}) we write
\be \label{CV3}
\Re[q_{\ve,m,L}(h(\cdot))- q_{\ve,m,L}(0)] \ = \ \Re[q_{\ve,m,L}(h(\cdot))- q_{\ve,m,L}(\Re h(\cdot))] +[q_{\ve,m,L}(\Re h(\cdot))- q_{\ve,m,L}(0)]  \ .
\ee
We have already established the bounds (\ref{L1}) for the second term on the RHS of (\ref{CV3}), so we focus on the first term. In view of the analyticity of the function $h(\cdot)\ra q_{\ve,m,L}(h(\cdot))$ 
we have that
\begin{multline} \label{CW3}
q_{\ve,m,L}(h(\cdot))-q_{\ve,m,L}(\Re h(\cdot)) \\ = \ \int_0^1d\al \ (1-\al)\frac{d^2}{d\al^2} q_{\ve,m,L}(\Re h(\cdot)+i\al \Im h(\cdot))  \
+ \frac{d}{d\al} q_{\ve,m,L}(\Re h(\cdot)+i\al \Im h(\cdot)) \Big|_{\al=0} \ , \\
\frac{d}{d\al} q_{\ve,m,L}(\Re h(\cdot)+i\al \Im h(\cdot)) \Big|_{\al=0} \ = \ iD_{\Im h}q_{\ve,m,L}(\Re h(\cdot)) \ , \\
\frac{d^2}{d\al^2} q_{\ve,m,L}(\Re h(\cdot)+i\al \Im h(\cdot)) \ = \ -D^2_{\Im h,  \Im h} q_{\ve,m,L}(\Re h(\cdot)+i\al \Im h(\cdot)) \  .
\end{multline}
It follows from (\ref{CW3}) that
\be \label{CX3}
\Re[q_{\ve,m,L}(h(\cdot))-q_{\ve,m,L}(\Re h(\cdot))] \ = \  -\int_0^1d\al \ (1-\al)\Re [D^2_{\Im h,  \Im h} q_{\ve,m,L}(\Re h(\cdot)+i\al \Im h(\cdot))] \  .
\ee
It is easy to see that for $h\in \ell_2(Q_L,\mathbb{C}^d)$ satisfying $\|\Im h(\cdot)\|_{2,L}<(\la-\eta)\del(\eta)$ the function $[a_1,a_2]\ra -\Re[D^2_{a_1,a_2}q _{\ve,m,L}( h(\cdot))], \ a_1,a_2\in\ell_2(Q_L,\mathbb{R}^d),$ is a quadratic form.  In the case of real $h(\cdot)$ we see from (\ref{P2}) that this quadratic form is positive definite, and from the BL inequality (\ref{G1}) that it is bounded above.  The inequality (\ref{W1}) will follow if we can  extend this upper bound to complex $h(\cdot)$. 

We consider $h(\cdot)\in  \ell_2(Q_L,\mathbb{C}^d)$ satisfying $\|\Im h(\cdot)\|_{2,L}<(\la-\eta)\del(\eta)$, and observe from (\ref{N2}),  Proposition 3.4 that
\be \label{CY3}
D_aq_{\ve,m,L}(h(\cdot)) \ = \  \lim_{T\ra\infty} \langle \ [a,\om^0_{\ve,m,h,L}(\cdot,T)]_L \rangle \ , \quad a\in\ell_2(Q_L,\mathbb{R}^d) \ .
\ee
We show that
\be \label{CZ3}
D^2_{a_1,a_2}q_{\ve,m,L}(h(\cdot)) \ = \  \lim_{T\ra\infty} \langle \ [a_1, D_{a_2}\om^0_{\ve,m,h,L}(\cdot,T)]_L \rangle \ , \quad a_1,a_2\in\ell_2(Q_L,\mathbb{R}^d) \ .
\ee
Differentiating with respect to $a\in \ell_2(Q_L,\mathbb{R}^d)$ the identity (\ref{BX3}) for $\om^0_{\ve,m,h,L}(\cdot,T)$, yields the formula
\begin{multline} \label{DA3}
D_a\om^0_{\ve,m,h,L}(\cdot,\cdot) \ = \ \mathcal{L}_TD_a\om^0_{\ve,m,h,L}(\cdot,\cdot)-k(\cdot,\cdot), \\
 k(\cdot,t)=\left\{1-e^{-t/2}\right\}\na[-\De+m^2]^{-1}\na^*a(\cdot) \ , \quad 0\le t\le T \  ,
\end{multline}
where $\mathcal{L}_T$ is defined by (\ref{O3}) with $\xi(\cdot)\equiv0$.  We have observed in the proof of Proposition 3.3 that $\|\mathcal{L}_T\|_{\mathcal{E}_T}\le (1-\la+\eta)<1$. Now (\ref{CZ3}) follows from (\ref{CY3}), (\ref{DA3}) by applying FTC to (\ref{BX3}). Since  $\|\mathcal{L}_T\|_{\mathcal{E}_T}<1$ the series expansion derived from (\ref{DA3}), 
\be \label{DB3}
D_a\om^0_{\ve,m,h,L}(\cdot,\cdot) \ = \  -\sum_{r=0}^\infty  \mathcal{L}_T^rk(\cdot,\cdot) \ ,
\ee
converges in the Banach space $\mathcal{E}_T$.  We also see that 
\be \label{DC3}
\left|  [a,      \mathcal{L}_T^rk(\cdot,T)]_L           \right| \ \le \ (1-\la+\eta)^r \| [-\De+m^2]^{-1/2}\na^*a\|^2_{2,L} \ ,  \quad r=0,1,\dots ,
\ee
whence we conclude that
\be \label{DD3}
\left|  \ [a, D_a\om^0_{\ve,m,h,L}(\cdot,T) ]_L \ \right| \ \le \ \frac{1}{\la-\eta}  \| [-\De+m^2]^{-1/2}\na^*a\|^2_{2,L} \ .
\ee
The inequality (\ref{W1})  now follows from (\ref{CV3}), (\ref{L1}),  (\ref{CX3}), (\ref{CZ3}), (\ref{DD3}). 

The proof of the lower bound (\ref{X1}) proceeds similarly, using (\ref{CX3}), (\ref{CZ3}).  Thus we need to establish a lower bound on $-\Re [D^2_{\Im h,  \Im h} q_{\ve,m,L}(\Re h(\cdot)+i\al \Im h(\cdot))]$
for $0<\al<1$,  which extends the lower bound (\ref{BT3}), (\ref{BU3}) applied to the SDE (\ref{B3}).  To do this we use the representation (\ref{BS3}), which applies for pure imaginary $\al$ with $|\al|<1$, and then analytically continue it to real $\al\in(0,1)$.  We first write (\ref{BS3}) in the case of odd $r=2m+1$ as
\begin{multline} \label{DE3}
[b,\mathcal{L}_{\ve,k,\infty,r}b]_n \ = \ 
    \left\|                    K^{1/2} A \left\{    \frac{1}{2} \int_0^\infty    dt \  \exp\left[  -t/2 +\mathcal{A}_kt\right]          KA  \right\}^m     b    \ \right\|_{\ve,k}^2    \ = \  \\
 \Bigg\langle \   \Bigg\|  \frac{1}{2^{m}}\int_{0<s_1<s_2<\cdots<s_{m}<\infty} ds_1\cdots ds_{m} \ e^{-s_m/2} \ \times \\
   K^{1/2}(\phi_{\ve,k}(T))E\left[AK(\phi_{\ve,k}(T+s_1))\cdots AK(\phi_{\ve,k}(T+s_{m}))Ab \ \big| \  \phi_{\ve,k}(T) \ \right] \ \Bigg\|_n^2 \ \Bigg\rangle_{\ve,k}  \ ,
\end{multline}
where $T\in\mathbb{R}$ is arbitrary. We may rewrite the final expression in (\ref{DE3}) as
\begin{multline} \label{DF3}
\left\langle \ \left[  F(\phi_{\ve,k}(T)), \ K(\phi_{\ve,k}(T)) F(\phi_{\ve,k}(T))\right]_n \ \right\rangle_{\ve,k} \ , \quad F_{\ve,k}(\phi) \ = \\
 \frac{1}{2^{m}}\int_{0<s_1<s_2<\cdots<s_{m}<\infty} ds_1\cdots ds_{m} \ e^{-s_m/2} \ E\left[AK(\phi_{\ve,k}(s_1))\cdots AK(\phi_{\ve,k}(s_{m}))Ab \ \big| \  \phi_{\ve,k}(0)=\phi \ \right]  \ .
\end{multline}
The analytic continuation of (\ref{DE3}) to complex $k$ may then be carried out by using the identity
\begin{multline} \label{DG3}
\lim_{T\ra\infty}\left\langle \ \left[  F_{\ve,k}(\phi_{\ve,k}(T)), \ K(\phi_{\ve,k}(T)) F_{\ve,k}(\phi_{\ve,k}(T))\right]_n \ \right\rangle_{\ve,k} \\
 = \ \lim_{T\ra\infty}E\left\{ \  \left[  F_{\ve,k}(\phi_{\ve,k}(T)), \ K(\phi_{\ve,k}(T)) F_{\ve,k}(\phi_{\ve,k}(T))\right]_n  \ \Big| \ \phi_{\ve,k}(0)=0 \ \right\} \  .
\end{multline}

In the case of the SDE (\ref{B3}) we have that $n=L^d$ and
\be \label{DH3}
A=[-\De+m^2]^{-1}, \quad k=\na^*[\Re h(\cdot)+i\al\Im h(\cdot)], \quad b=\na^* \Im h(\cdot) \ , \quad K(\phi)=\na^*\mathbf{b}(\na\phi(\cdot))\na \ .
\ee
Furthermore we have that
\be \label{DI3}
\|\na F_{\ve,k}(\cdot)\|_{2,L} \ \le \ (1-\la+\eta)^m\left\|  (-\De+m^2)^{-1/2}\na^*\Im h(\cdot)           \right\|_{2,L} \quad {\rm  for \ } |\al|<1 \  .  
\ee
We define a function $\phi\ra G_{\ve,k}(\phi)$ similar to $F_{\ve,k}(\cdot)$ by
\begin{multline} \label{DJ3}
G_{\ve,k}(\phi) \ = \\
 \frac{1}{2^{m}}\int_{0<s_1<s_2<\cdots<s_{m}<\infty} ds_1\cdots ds_{m} \ e^{-s_m/2} \ E\left[AK(\Re \phi_{\ve,k}(s_1))\cdots AK(\Re\phi_{\ve,k}(s_{m}))Ab \ \big| \  \phi_{\ve,k}(0)=\phi \ \right]  \ .
\end{multline}
It follows from (\ref{O1}) that
\be \label{DK3}
\|\na F_{\ve,k}(\cdot)-\na G_{\ve,k}(\cdot)\|_{2,L} \ \le \ \eta m(1-\la+\eta)^{m-1}\left\|  (-\De+m^2)^{-1/2}\na^*\Im h(\cdot)           \right\|_{2,L} \quad {\rm  for \ } |\al|<1 \  .
\ee
Writing
\begin{multline} \label{DL3}
  \left[  F_{\ve,k}(\phi_{\ve,k}(T)), \ K(\phi_{\ve,k}(T)) F_{\ve,k}(\phi_{\ve,k}(T))\right]_n  \ = \  \left[  F_{\ve,k}(\phi_{\ve,k}(T)), \ \{K(\phi_{\ve,k}(T))- K(\Re\phi_{\ve,k}(T))\}F_{\ve,k}(\phi_{\ve,k}(T))\right]_n \\
  +  \left[  F_{\ve,k}(\phi_{\ve,k}(T))-G_{\ve,k}(\phi_{\ve,k}(T)), \ K(\Re\phi_{\ve,k}(T)) F_{\ve,k}(\phi_{\ve,k}(T))\right]_n \\
   +  \left[  F_{\ve,k}(\phi_{\ve,k}(T))-G_{\ve,k}(\phi_{\ve,k}(T)), \ K(\Re\phi_{\ve,k}(T)) G_{\ve,k}(\phi_{\ve,k}(T))\right]_n 
   + \left[  G_{\ve,k}(\phi_{\ve,k}(T)), \ K(\Re\phi_{\ve,k}(T)) G_{\ve,k}(\phi_{\ve,k}(T))\right]_n  \ ,
\end{multline}
and noting that since $G_{\ve,k}(\cdot)$ is real valued the final term in (\ref{DL3}) is non-negative, we conclude from  (\ref{DI3}), (\ref{DK3}), (\ref{DL3}) that
\be \label{DM3}
\Re  \left[  F_{\ve,k}(\phi_{\ve,k}(T)), \ K(\phi_{\ve,k}(T)) F_{\ve,k}(\phi_{\ve,k}(T))\right]_n   \ \ge \ 
-\eta(2m+1)(1-\la+\eta)^{2m} \left\|  (-\De+m^2)^{-1/2}\na^*\Im h(\cdot)           \right\|_{2,L}^2 \ .
\ee

A similar argument can be made for (\ref{BS3}) in the case of even $r=2m$. Observing that
\be \label{DN3}
(-2\mathcal{A}_k+I_n)^{-1/2} \ = \ \frac{1}{\sqrt{2\pi}}\int_0^\infty dt \  t^{-1/2} \exp[-t/2+\mathcal{A}_kt] \ ,
\ee
Then we have that
\begin{multline} \label{DO3}
[b,\mathcal{L}_{\ve,k,\infty,r}b]_n \ = \ \left\langle \ \left[  F_{\ve,k}(\phi_{\ve,k}(T)), \ F_{\ve,k}(\phi_{\ve,k}(T))\right]_n  \ \right\rangle_{\ve,k}  \ , \quad F_{\ve,k}(\phi) \ = \\
\frac{1}{2^{m-1}\sqrt{2\pi}}\int_{0<s_1<s_2<\cdots<s_{m}<\infty} ds_1\cdots ds_{m} \ s_1^{-1/2}e^{-s_m/2} \times \\
 E\left[A^{1/2}K(\phi_{\ve,k}(s_1))AK(\phi_{\ve,k}(s_{2}))\cdots AK(\phi_{\ve,k}(s_{m}))Ab \ \big| \  \phi_{\ve,k}(0)=\phi \ \right]  \ .
\end{multline}
Letting $G_{\ve,k}(\cdot)$ be defined as $F_{\ve,k}(\cdot)$ in (\ref{DO3}), but with $\phi_{\ve,k}(\cdot)$ replaced by  $\Re\phi_{\ve,k}(\cdot)$, we see that
\be \label{DP3}
\begin{array}{lcl}
\| F_{\ve,k}(\cdot)\|_{2,L} \ &\le& \ (1-\la+\eta)^m\left\|  (-\De+m^2)^{-1/2}\na^*\Im h(\cdot)           \right\|_{2,L} \ ,  \\
\|F_{\ve,k}(\cdot)- G_{\ve,k}(\cdot)\|_{2,L} \ &\le& \ \eta m(1-\la+\eta)^{m-1}\left\|  (-\De+m^2)^{-1/2}\na^*\Im h(\cdot)           \right\|_{2,L}   .
\end{array}
\ee
Similarly to (\ref{DM3}) we conclude from (\ref{DP3}) that
\be \label{DQ3}
\Re  \left[  F_{\ve,k}(\phi_{\ve,k}(T)), \  F_{\ve,k}(\phi_{\ve,k}(T))\right]_n   \ \ge \ 
-2m\eta(1-\la+\eta)^{2m-1} \left\|  (-\De+m^2)^{-1/2}\na^*\Im h(\cdot)           \right\|_{2,L}^2 \ .
\ee

We have now from (\ref{BI3}), (\ref{BT3}), (\ref{CZ3}), (\ref{DM3}), (\ref{DQ3}) that
\begin{multline} \label{DR3}
\Re [D^2_{\Im h,  \Im h} q_{\ve,m,L}(\Re h(\cdot)+i\al \Im h(\cdot))] \ \ge \\
\left[1-   \eta\sum_{r=1}^\infty r(1-\la+\eta)^{r-1}                \right] \left\|  (-\De+m^2)^{1/2}\na^*\Im h(\cdot)           \right\|_{2,L}^2 \\
= \ \left[1-\frac{\eta}{(\la-\eta)^2}\right] \left\|  (-\De+m^2)^{1/2}\na^*\Im h(\cdot)           \right\|_{2,L}^2   \ .
\end{multline}
The inequality (\ref{X1})  follows from (\ref{CX3}), (\ref{DR3}). 
\end{proof}

\vspace{.1in}
\section{Higher order derivatives and $\ell_p$ theory}
In $\S3$ we have seen how the inequalities (\ref{N2})-(\ref{Q2}) on the derivatives $D_a q_{\ve,m,L}(h(\cdot))$, $D^2_{a_1,a_2} q_{\ve,m,L}(h(\cdot))$ may be extended to complex $h(\cdot)$ for $V(\cdot)$ satisfying (\ref{D1}), (\ref{O1}).  Thus we have from (\ref{CA3}), (\ref{CY3}) that
\be \label{A4}
|D_a q_{\ve,m,L}(h(\cdot))| \ \le \     \frac{1}{\la-\eta}\|(-\De+m^2)^{-1/2}\na^*a\|_{2,L}\|(-\De+m^2)^{-1/2}\na^*h\|_{2,L} \ ,
\ee
for $a,h\in\ell_2(Q_L,\mathbb{C}^d)$  and $h(\cdot)$ satisfying the inequality $\|\Im h(\cdot)\|_{2,L}< (\la-\eta)\del(\eta)$.
Similarly we have from (\ref{CZ3}) and the argument following it that
\be \label{B4}
|D^2_{a_1,a_2} q_{\ve,m,L}(h(\cdot))| \ \le \     \frac{1}{\la-\eta}\|(-\De+m^2)^{-1/2}\na^*a_1\|_{2,L}\|(-\De+m^2)^{-1/2}\na^*a_2\|_{2,L} \ ,
\ee
$a_1,a_2,h\in\ell_2(Q_L,\mathbb{C}^d)$  and $h(\cdot)$ satisfying the inequality $\|\Im h(\cdot)\|_{2,L}< (\la-\eta)\del(\eta)$.
We can use the same methodology to obtain bounds on higher order directional derivatives of the function $h(\cdot)\ra q_{\ve,m,L}(h(\cdot))$.

To obtain a bound on the third directional derivative  of $h(\cdot)\ra q_{\ve,m,L}(h(\cdot))$ when $h(\cdot)$ is real we differentiate (\ref{CZ3}), yielding  the formula
\be \label{C4}
D^3_{a_1,a_2,a_3} q_{\ve,m,L}(h(\cdot)) \ = \ \lim_{T\ra\infty} \left\langle \ [a_1(\cdot), D^2_{a_2,a_3}\om^0_{\ve,m,h,L}(\cdot,T)]_L \ \right\rangle \ .
\ee
We may obtain an  equation for $g(\cdot,t)= D^2_{a_2,a_3}\om^0_{\ve,m,h,L}(\cdot,t)$  by differentiating (\ref{DA3}). This yields the equation 
\be \label{D4}
g(\cdot,t) \ = \  \mathcal{L}_Tg(\cdot,t) -k(\cdot,t) \ , \\
\ee
where $k(\cdot,t)\in\ell_2(Q_L,\mathbb{C}^d)$ is given by the formula
\begin{multline} \label{E4}
[a(\cdot),k(\cdot,t)]_L \ = \ \frac{1}{2} \int_0^te^{-(t-s)/2} \ ds \ \times  \ \sum_{x\in Q_L} \\
 V'''\left(         \om^0_{\ve,m,h,L}(x,s)                 \right)\left[ \na[-\De+m^2]^{-1}\na^*a(x),    D_{a_2}\om^0_{\ve,m,h,L}(x,s),          D_{a_3}\om^0_{\ve,m,h,L}(x,s)       \right] \ , \\
  {\rm for \ all \  }  a(\cdot)\in\ell_2(Q_L,\mathbb{R}^d) \ .
\end{multline}
We see on applying the H\"{o}lder inequality in (\ref{E4}), using (\ref{Y1}), (\ref{Z1}) and the bound on the solution to (\ref{DA3}), that
\begin{multline} \label{F4}
|[a(\cdot),k(\cdot,t)]_L| \ \le  \  M\|a(\cdot)\|_{2,L}  \sup_{s>0}\|D_{a_2}\om^0_{\ve,m,h,L}(\cdot,s)\|_{4,L}\sup_{s>0}\|D_{a_3}\om^0_{\ve,m,h,L}(\cdot,s)\|_{4,L} \\
\le \ M\|a(\cdot)\|_{2,L}  \sup_{s>0}\|D_{a_2}\om^0_{\ve,m,h,L}(\cdot,s)\|_{2,L}\sup_{s>0}\|D_{a_3}\om^0_{\ve,m,h,L}(\cdot,s)\|_{2,L}  \\
\le \la^{-2}  M\|a(\cdot)\|_{2,L} \|a_2(\cdot)\|_{2,L} \|a_3(\cdot)\|_{2,L}  \ , \quad a(\cdot)\in\ell_2(Q_L,\mathbb{R}^d) \ .
\end{multline}
It follows from (\ref{D4}), (\ref{F4}) that
\be  \label{G4}
\begin{array}{lcr}
\|k(\cdot,t)\|_{2,L} \ &\le& \ \la^{-2} M \|a_2(\cdot)\|_{2,L} \|a_3(\cdot)\|_{2,L} \ , \quad t>0, \\
\|g(\cdot,t)\|_{2,L} \ &\le &\ \la^{-3}  M \|a_2(\cdot)\|_{2,L} \|a_3(\cdot)\|_{2,L} \ , \quad t>0 \ .
\end{array}
\ee
We conclude from (\ref{C4}), (\ref{G4}) that
\be \label{H4}
|D^3_{a_1,a_2,a_3} q_{\ve,m,L}(h(\cdot))| \ \le \  \la^{-3}M  \|a_1(\cdot)\|_{2,L}\|a_2(\cdot)\|_{2,L} \|a_3(\cdot)\|_{2,L}  \ .
\ee
The inequality (\ref{AA1}) with $C(\la)=1/6\la^3$ follows from (\ref{H4}) and the identity
\begin{multline} \label{I4}
q_{\ve,m,L}(h(\cdot))-q_{\ve,m,L}(0)+\frac{1}{2\ve}\left\langle \ [h(\cdot),\na\phi(\cdot)]_L^2 \ \right\rangle_{\ve,m,0,L}  \\ 
= \ q_{\ve,m,L}(h(\cdot))-q_{\ve,m,L}(0)-\frac{d}{d\al} q_{\ve,m,L}(\al h(\cdot))\Big|_{\al=0} -\frac{1}{2}\frac{d^2}{d\al^2} q_{\ve,m,L}(\al h(\cdot))\Big|_{\al=0}  \\
 = \ \frac{1}{2}\int_0^1 d\al \ (1-\al)^2 \frac{d^3}{d\al^3}  q_{\ve,m,L}(\al h(\cdot)) \ = \  \frac{1}{2}\int_0^1 d\al \ (1-\al)^2 D^3_{h,h,h} q_{\ve,m,L}(\al h(\cdot))\ .
\end{multline}

The bounds we have obtained so far on the first three directional derivatives of the function $h(\cdot)\ra q_{\ve,m,L}(h(\cdot))$ are in terms of  Euclidean  $\ell_2$ norms. We may generalize these inequalities to inequalities involving $\ell_p$ norms with $p>1$ by applying the CZ theorem \cite{stein} to functions on $Q_L$, as we did in the proof of Theorem 1.1. First we note the following extension of Lemma 3.1 to $\ell_p$:
\begin{lem} Assume the function $V(\cdot)$ is holomorphic in a strip parallel to the real axis and satisfies the inequalities (\ref{D1}), (\ref{O1}).  Let $p>1$ and $\kappa_p$ be the constant in (\ref{R2}). Then
for all   $\eta>0$ satisfying $\la-\eta>1-1/\kappa_p$ and  $h,\xi\in \ell_2(Q_L,\mathbb{C}^d)$  with $\kappa_p\|\Im h\|_{p,L}, \ \|\Im \xi\|_{p,L}<[\kappa_p(\la-\eta)-(\kappa_p-1)]\del(\eta)$, the SDE (\ref{C3}) with initial condition $\xi$  has a  unique strong solution $t\ra \om^\xi_{\ve,m,h,L}(\cdot,t), \ t>0,$  globally in time. Furthermore, the inequality
$\sup_{t>0}\|\Im \om^\xi_{\ve,m,h,L}(\cdot,t)\|_{p,L}< \del(\eta)$ holds.
\end{lem}
\begin{proof}
The main point is that the inequality (\ref{BY3}) may be generalized to
\begin{multline} \label{J4}
\|\Im \mathcal{K}_Tf(\cdot,t)\|_{p,L} \ \le \ e^{-t/2}\|\Im\xi\|_{p,L} +\left\{1-e^{-t/2}\right\}\kappa_p\|\Im h\|_{p,L} \\
+ \ \frac{(1-\la+\eta)\kappa_p}{2} \int_0^te^{-(t-s)/2}\|\Im f(\cdot,s\|_{p,L} \ ds  \ , \quad 0<t<T \ ,
\end{multline}
\end{proof}
Next we obtain from Lemma 4.1 an extension of Proposition 3.4 to $\ell_p$:
\begin{proposition}
Assume the function $V(\cdot)$ satisfies the conditions of Lemma 4.1 and $p>1$. Let $g_{\ve,m,L}:\ell_2(Q_L,\mathbb{R}^d)\ra \ell_2(Q_L,\mathbb{R}^d)$ be defined by $g_{\ve,m,L}(h(\cdot))= \left\langle \  \na\phi(\cdot) \ \right\rangle_{\ve,m,h,L}$. Then $g_{\ve,m,L}$ extends analytically  to the strip $\{h\in \ell_p(Q_L,\mathbb{C}^d): \ \|\Im h(\cdot)\|_{p,L}<[(\la-\eta)-(1-1/\kappa_p)]\del(\eta)\}$ with the holomorphic function  
$g_{\ve,m,L}(\cdot)$ taking values in  $\ell_2(Q_L,\mathbb{C}^d)$.  Furthermore, one has that
\be \label{K4}
\lim_{T\ra\infty}\left\|\left\langle \  \om^0_{\ve,m,h,L}(\cdot,T) \ \right\rangle-g_{\ve,m,L}(h(\cdot)) \ \right\|_{2,L} \ =  0  , \ 
\ee
with uniform convergence in any region $\{h\in \ell_p(Q_L,\mathbb{C}^d): \ \|\Re h(\cdot)\|_{p,L}\le M, \ \|\Im h(\cdot)\|_{p,L}\le [(\la-\eta)-(1-1/\kappa_p)]\del\}, \ M>0, \ 0<\del<\del(\eta)$.
\end{proposition}
\begin{proof}
The main point to observe is that from Lemma 4.1 we have the inequality $\sup_{t>0,x\in Q_L}|\Im \om^0_{\ve,m,h,L}(x,t)|\le \sup_{t>0}\|\Im \om^0_{\ve,m,h,L}(\cdot,t)\|_{p,L}< \del(\eta)$. Then we follow the proof of Proposition 3.4.
\end{proof}
\begin{proposition}
Assume $V(\cdot)$  satisfies the conditions of Lemma 4.1  and $p>1$.  Then the function  $h(\cdot)\ra q_{\ve,m,L}(h(\cdot))$,  defined for real $h\in\ell_p(Q_L,\mathbb{R}^d)$, extends analytically to complex $h\in \ell_p(Q_L,\mathbb{C}^d)$ satisfying $\|\Im h(\cdot)\|_{p,L}<[\la-\eta-(1-1/\kappa_p)]\del(\eta)$.  Furthermore, if $p'=p/(p-1)$ is the conjugate of $p$ then the following inequality holds:
\be \label{L4}
|D_aq_{\ve,m,L}(h(\cdot))| \ \le \ \frac{1}{[(\la-\eta)-(1-1/\kappa_p)]} \|a\|_{p',L}\|h\|_{p,L} \ , \quad  {\rm for \ } a\in\ell_{p'}(Q_L,\mathbb{C}^d) \ .
\ee
Let $q,q'>1$ be conjugate so  $1/q+1/q'=1$. Then second derivatives of the function $h(\cdot)\ra q_{\ve,m,L}(h(\cdot))$ satisfy the inequality
\be \label{M4}
|D^2_{a_1,a_2}q_{\ve,m,L}(h(\cdot))| \ \le \   \frac{1}{[(\la-\eta)-(1-1/\kappa_q)]}\|a_1\|_{q',L}\|a_2\|_{q,L}  \ , \quad  {\rm for \ } a_1,a_2:Q_L\ra\mathbb{C}^d \ .
\ee

Assume $V(\cdot)$, in addition to satisfying (\ref{D1}), (\ref{O1}), also satisfies (\ref{AB1}), and $q_1,q_2,q_3>1$ satisfy $1/q_1+1/q_2+1/q_3\ge1$. Then
\be \label{N4}
|D^3_{a_1,a_2,a_3}q_{\ve,m,L}(h(\cdot))| \ \le \   \frac{M_\eta}{\prod_{j=1}^3[(\la-\eta)-(1-1/\kappa_{q_j})]}\prod_{j=1}^3\|a_j\|_{q_j,L} \ , \quad  {\rm for \ } a_j:Q_L\ra\mathbb{C}^d, \ j=1,2,3 .
\ee
\end{proposition}
\begin{proof}
We conclude from (\ref{CU3}) and  Proposition 4.1 that  the function  $h(\cdot)\ra q_{\ve,m,L}(h(\cdot))$  defined for real $h\in\ell_p(Q_L,\mathbb{R}^d)$ extends analytically to complex $h\in \ell_p(Q_L,\mathbb{C}^d)$ satisfying $\|\Im h(\cdot)\|_{p,L}<[(\la-\eta)-(1-1/\kappa_p)]\del(\eta)]$.   To prove (\ref{L4}) we proceed similarly to the proof of (\ref{A4}).   The operator
$\mathcal{L}_{T,\al}$ in (\ref{Y3}) acts on the Banach space $\mathcal{E}_{T,p}$ of continuous functions  $f:Q_L\times [0,T]\ra\mathbb{C}^d$ with norm $\|f(\cdot,\cdot)\|_{\mathcal{E}_{T,p}}=\sup_{0\le t\le T}\|f(\cdot,t)\|_{p,L}$ has norm $\|\mathcal{L}_T\|_{\mathcal{E}_{T,p}}\le \kappa_p\{1-\la+\eta\}$.  Hence the solution to (\ref{Y3}) satisfies the inequality $\|f_\al\|_{\mathcal{E}_{T,p}}\le\|h\|_{p,L}/[(\la-\eta)-(1-1/\kappa_p)]$. The proof of (\ref{M4}) proceeds similarly on using the  equation (\ref{DA3}).

To prove (\ref{N4}) we observe that similarly to (\ref{F4}) we have the inequality
\begin{multline} \label{O4}
|[a(\cdot),k(\cdot,t)]_L| \ \le  \  M_\eta\kappa_{q_1}\|a(\cdot)\|_{q_1,L}  \sup_{s>0}\|D_{a_2}\om^0_{\ve,m,h,L}(\cdot,s)\|_{q_2,L}\sup_{s>0}\|D_{a_3}\om^0_{\ve,m,h,L}(\cdot,s)\|_{q_3,L} \\
\le \ \frac{ M_\eta\kappa_{q_1}\|a(\cdot)\|_{q_1,L}  }{\prod_{j=2}^3[(\la-\eta)-(1-1/\kappa_{q_j})]}\prod_{j=2}^3\|a_j\|_{q_j,L}  \ ,
\end{multline}
where $M_\eta$ is the constant of (\ref{AB1}).
It follows from (\ref{O4}) that if $1/q_1+1/q'_1=1$ then
\be \label{P4}
\|k(\cdot,t)\|_{q'_1,L} \ \le \ \frac{ M_\eta\kappa_{q_1}  }{\prod_{j=2}^3[(\la-\eta)-(1-1/\kappa_{q_j})]}\prod_{j=2}^3\|a_j\|_{q_j,L}  \ .
\ee
The solution to (\ref{D4}) satisfies the inequality
\be \label{Q4}
\|g(\cdot,t)\|_{q'_1,L} \ \le \ \frac{1}{\kappa_{q'_1}[(\la-\eta)-(1-1/\kappa_{q'_1})]} \sup_{s>0}\|k(\cdot,s)\|_{q'_1,L} \ .
\ee
Now (\ref{N4}) follows from (\ref{P4}), (\ref{Q4}) on using the identity $\kappa_q=\kappa_{q'}$ for all conjugate $q,q'>1$. 
\end{proof}
\begin{proof}[Proof of Theorem 1.3]
We use the identity (\ref{I4}) and estimate the RHS using (\ref{N4}) with $a_1=a_2=a_3=h$ and $q_1=q_2=q_3=p$. 
\end{proof}
\begin{proof}[Proof of Theorem 1.4]
We first consider the case $d\ge 3$. For  two functions $h,h':Q_L\ra\mathbb{C}^d$ we may from (\ref{I1}) represent the covariance of  the
variables  $\exp\{-[h(\cdot),\na\phi(\cdot)]_L/\ve\}$ and  $\exp\{-[h'(\cdot),\na\phi(\cdot)]_L/\ve\}$  with respect to the measure 
$\langle\cdot\rangle_{\ve,m,0,L}$ in terms of the function $q_{\ve,m,L}(\cdot)$ by
\begin{multline} \label{R4}
{\rm cov}_{\ve,m,0,L}\left\{   \exp\left[   -\frac{ [h(\cdot),\na\phi(\cdot)]_L }{\ve}  \right] ,  \  \exp\left[   -\frac{ [h'(\cdot),\na\phi(\cdot)]_L }{\ve}  \right] \ \ \right\}  \ = \\
\left\langle \  \exp\left[   -\frac{ [h(\cdot),\na\phi(\cdot)]_L }{\ve}  \right] \ \ \right\rangle_{\ve,m,0,L}  \ \left\langle \  \exp\left[   -\frac{ [h'(\cdot),\na\phi(\cdot)]_L }{\ve}  \right] \ \ \right\rangle_{\ve,m,0,L}
 \ \times \\
  \left\{\exp\left[-\frac{q_{\ve,m,L}(h(\cdot)+h'(\cdot))-q_{\ve,m,L}(h(\cdot))-q_{\ve,m,L}(h'(\cdot))+q_{\ve,m,L}(0)}{\ve}\right]-1\right\}  \ .
\end{multline}
From Taylor's theorem we have that
\begin{multline} \label{S4}
q_{\ve,m,L}(h(\cdot)+h'(\cdot))-q_{\ve,m,L}(h(\cdot))-q_{\ve,m,L}(h'(\cdot))+q_{\ve,m,L}(0) \\
= \ \int_0^1\int_0^1 d\al  \ d\beta  \ D^2_{h,h'} q_{\ve,m,L}(\al h(\cdot)+\beta h'(\cdot))  \\
 = \  D^2_{h,h'} q_{\ve,m,L}(0)+\int_0^1\int_0^1\int_0^1 d\al  \ d\beta  \ d\ga \ D^3_{h,h',\al h+\beta h'} q_{\ve,m,L}(\ga[\al h(\cdot)+\beta h'(\cdot)]) \ \ .
\end{multline}

We choose $h(\cdot),h'(\cdot)$ in (\ref{R4}), depending on $L$, so that the limit of (\ref{R4}) as $L\ra\infty,m\ra0$ yields the covariance in  (\ref{AN1}).  To do this we note that the periodic 
Green's function $G_{\nu,L}(\cdot)$  on $Q_L$ corresponding to the Green's function $G_\nu(\cdot)$ on $\mathbb{Z}^d$ defined by (\ref{AE1}) is given by 
\be \label{T4}
G_{\nu,L}(x) \ = \ \sum_{n\in\mathbb{Z}^d}G_\nu(x+Ln)\  , \quad x\in Q_L \ .
\ee
For $x\in Q_L$ we define $h_{x,\nu,L}:Q_L\ra\mathbb{R}^d$ by $h_{x,\nu,L}(y)=\na G_{\nu,L}(y-x), \ y\in Q_L$.  We then have for any $\rho\in\mathbb{C}$ that
\begin{multline} \label{U4}
\lim_{m\ra0}\lim_{L\ra\infty} {\rm cov}_{\ve,m,0,L}\left\{   \exp\left[   \frac{\rho [h_{x,\nu,L}(\cdot),\na\phi(\cdot)]_L }{\ve}  \right] ,  \  \exp\left[   -\frac{ \rho[h_{0,\nu,L}(\cdot),\na\phi(\cdot)]_L }{\ve}  \right] \ \ \right\}  \\
= \  {\rm cov}_\ve\left\{   \exp\left[   \frac{\rho [h_{x,\nu}(\cdot),\na\phi(\cdot)]_\infty }{\ve}  \right] ,  \  \exp\left[   -\frac{ \rho[h_{0,\nu}(\cdot),\na\phi(\cdot)]_\infty }{\ve}  \right] \ \ \right\}  \ .
\end{multline}
To see (\ref{U4}) we use the identity
\begin{multline} \label{V4}
 \left\langle \  \exp\left[   -\frac{ [h(\cdot)+a(\cdot),\na\phi(\cdot)]_L }{\ve}  \right] \ \ \right\rangle_{\ve,m,0,L}- \left\langle \  \exp\left[   -\frac{ [h(\cdot),\na\phi(\cdot)]_L }{\ve}  \right] \ \ \right\rangle_{\ve,m,0,L} \\
  = \ -\frac{1}{\ve}\int_0^1d\al \ \left\langle  [a(\cdot),\na\phi(\cdot)]   \exp\left\{   -\frac{1}{\ve}[h(\cdot)+\al a(\cdot),\na\phi(\cdot)]   \right\}       \right\rangle_{\ve,m,0,L}  \ ,
  \quad h(\cdot),a(\cdot):Q_L\ra\mathbb{C}^d \ .
\end{multline}
 Applying the Schwarz inequality to the RHS of (\ref{V4}) followed by the BL inequality, the limit (\ref{U4}) is a consequence of  (\ref{AD1}).  In view of (\ref{AF1}), if we take the limit  as $\nu\ra0$ of the covariance on the RHS of (\ref{U4})  we obtain the covariance in (\ref{AN1}). 
 We have from (\ref{P2}) that the first term on the RHS of (\ref{S4}) is
\be \label{W4}
D^2_{h,h'} q_{\ve,m,L}(0) \ = \ -\ve^{-1} \left\langle   \    [h(\cdot),\na\phi]_L, \ [h'(\cdot),\na\phi(\cdot)]_L \                         \right\rangle_{\ve,m,0,L} \ .
\ee
Choosing $h=-\rho h_{x,\nu,L},  \ h'=\rho h_{0,\nu,L}$,  and taking the limits $L\ra\infty,m\ra0,\nu\ra0$ on the RHS of (\ref{W4})  we see similarly to  before that the RHS of (\ref{W4}) converges to  
$\ve^{-1}\rho^2\langle \phi(x)\phi(0)\rangle_\ve$. 

We bound the third derivative term on the RHS of (\ref{S4}) when  $h=-\rho h_{x,\nu,L},  \ h'=\rho h_{0,\nu,L}$, uniformly as $L\ra\infty,m\ra0,\nu\ra0$.  To do this we use the identity (\ref{C4}).  Comparing  (\ref{D4}), (\ref{E4}) to  (\ref{P3}), we see from (\ref{S3}) that
\begin{multline} \label{X4}
D^3_{a_1,a_2,a_3} q_{\ve,m,L}(h(\cdot)) \ = \ -\lim_{T\ra\infty} \left\langle \ [a_1(\cdot), \{I-\mathcal{L}_T\}^{-1} k(\cdot,T)]_L \ \right\rangle \\
= \ -\lim_{T\ra\infty}\left\{ \left\langle \  \lim_{n\ra\infty} \frac{1}{2}\int_0^Te^{-(T-t)/2} [a_n(\cdot,t,T),g(\cdot,t)]_L  \ dt \ \right\rangle  \ \right\}  \ ,
\end{multline}
where $g(\cdot,t)$ is defined from(\ref{E4})  by
\begin{multline} \label{Y4}
  {\rm For \ all \  }  a(\cdot)\in\ell_2(Q_L,\mathbb{R}^d) \ , \quad  \quad [a(\cdot),g(\cdot,t)]_L \ = \\   \sum_{y\in Q_L}  \
 V'''\left(         \om^0_{\ve,m,h,L}(y,t)                 \right)\left[ \na[-\De+m^2]^{-1}\na^*a(y),    D_{a_2}\om^0_{\ve,m,h,L}(y,t),          D_{a_3}\om^0_{\ve,m,h,L}(y,t)       \right] \ ,
\end{multline}
and $a_n(\cdot,t,T)$ is given by (\ref{T3}) with $a(\cdot)\equiv a_1(\cdot)$.  We conclude from (\ref{X4}), (\ref{Y4}) that
\begin{multline} \label{Z4}
D^3_{a_1,a_2,a_3} q_{\ve,m,L}(h(\cdot)) \ = \ - \lim_{T\ra\infty}\lim_{n\ra\infty} \frac{1}{2}\int_0^T dt \ e^{-(T-t)/2}\sum_{y\in Q_L}  \\
 V'''\left(         \om^0_{\ve,m,h,L}(y,t)                 \right)\left[ \na[-\De+m^2]^{-1}\na^*a_n(y,t,T),    D_{a_2}\om^0_{\ve,m,h,L}(y,t),          D_{a_3}\om^0_{\ve,m,h,L}(y,t)       \right]  \ .
\end{multline}

For the purposes of estimating the RHS of (\ref{S4}) we take $a_1=-\rho h_{x,\nu,L},  \ a_2=\rho h_{0,\nu,L},  \ a_3=\al a_1+\beta a_2, \ h=\ga a_3$.  Next observe that we may consider the operator 
$\na[-\De+m^2]^{-1}\na^*$, which occurs in the expression (\ref{Z4}),  as acting on functions $a:\mathbb{Z}^d\ra\mathbb{C}^d$. In the case of (\ref{Z4}) these functions are periodic on $\mathbb{Z}^d$  with $Q_L$ as their fundamental domain. Similarly  the random function $ \om^0_{\ve,m,h,L}(\cdot,t)$ is periodic on $\mathbb{Z}^d$.  Hence the operator $\mathcal{L}_T$ of (\ref{O3}) is simply an operator  on periodic functions on $\mathbb{Z}^d$, which we may then extend to an operator on non-periodic functions.  Thus we can solve (\ref{DA3}) with $a=\rho h_{0,\nu,L}$ replaced by $a=\rho h_{0,\nu}$.
Carrying this out in all the terms that occur in (\ref{Z4}), we obtain a function $D^{3,{\rm approx}}_{a_1,a_2,a_3} q_{\ve,m,L}(h(\cdot))$. It is easy to see that
\be \label{AA4}
\lim_{L\ra\infty} \left[D^3_{a_1,a_2,a_3} q_{\ve,m,L}(h(\cdot))-D^{3,{\rm approx}}_{a_1,a_2,a_3} q_{\ve,m,L}(h(\cdot))\right] \ = \ 0 \ .
\ee 

We estimate $D^{3,{\rm approx}}_{a_1,a_2,a_3} q_{\ve,m,L}(h(\cdot))$ with $a_1=-\rho h_{x,\nu},  \ a_2=\rho h_{0,\nu},  \ a_3=\al a_1+\beta a_2, \ h=\ga a_3$, using weighted norm inequalities for singular integrals. For a function $a:\mathbb{Z}^d\ra\mathbb{C}^d$  we define the weighted $p$ norm of $a(\cdot)$ 
with weight $w(\cdot)$ by
\be \label{AB4}
\|a\|_{p,w} \ = \ \left[\sum_{y\in\mathbb{Z}^d} |a(y)|^pw(y) \ \right]^{1/p} \ ,
\ee
whence we have that
\be \label{AC4}
|a(y)| \ \le \ \frac{\|a\|_{p,w}}{w(y)^{1/p}} \ , \quad y\in\mathbb{Z}^d \ .
\ee
It is well known there exists a constant $C_d$, independent of $\nu>0$, such that
\be \label{AD4}
|h_{0,\nu}(y)|\le \frac{C_d}{|y|^{d-1}+1} \ , \quad |h_{x,\nu}(y)|\le \frac{C_d}{|x-y|^{d-1}+1} \ , \quad x,y\in\mathbb{Z}^d \ .
\ee
For $\al,p$ satisfying $0<\al<1/2, \ 1<p<\infty$, we define a weight $w_{\al,p}(\cdot)$ on $\mathbb{Z}^d$ by $w_{\al,p}(y)=[1+|y|^{d-1-\al}]^p, \ y\in\mathbb{Z}^d$.  We have then from (\ref{L3}) that 
$\|h_{0,\nu}\|_{p,w_{\al,p}}\le C_{\al,p}$ for a constant independent of $\nu>0$ if $p\al>d$.  Furthermore, if $p\al>d$ then $w_{\al,p}(\cdot)$ also satisfies the discrete Muckenhoupt $A_p$ condition \cite{stein1} on  $\mathbb{Z}^d$.  It follows  there is a constant $\kappa_{p,\al}$, independent of $m>0$,  such that
\be \label{AE4}
\|\na(-\De+m^2)^{-1}\na^*a(\cdot)\|_{p,w_{\al,p}} \ \le \  \kappa_{p,\al} \|a(\cdot)\|_{p,w_{\al,p}} \ , \quad a:\mathbb{Z}^d\ra\mathbb{C}^d \ . 
\ee

We assume now that  $\rho\in\mathbb{C}$ satisfies the inequality $|\Im\rho|\le (\la-\eta)\del(\eta)/\sup_{0<\nu<1}\|h_{0,\nu}\|_{2,\infty}$, so that the conditions of Lemma 3.1 are satisfied for sufficiently large $L$.  It follows from (\ref{AE4}) that if $\kappa_{p,\al}(1-\la+\eta)<1$ then the solution $ D_{a_2}\om^{0,{\rm approx}}_{\ve,m,h,L}(y,\cdot), \ y\in\mathbb{Z}^d,$  of (\ref{DA3}) with $a=\rho h_{0,\nu}$ is bounded  by
\be \label{AF4}
\sup_{0<t\le T} \left\|    D_{a_2}\om^{0,{\rm approx}}_{\ve,m,h,L}(\cdot,t)               \right\|_{p,w_{\al,p}} \ \le \ \frac{\rho\|h_{0,\nu}\|_{p,w_{\al,p}}}{\la-\eta-(1-1/\kappa_{p,\al})} \ .
\ee
We conclude from (\ref{AC4})-(\ref{AF4}) that
\be \label{AG4}
\sup_{0<t\le T} \left|    D_{a_2}\om^{0,{\rm approx}}_{\ve,m,h,L}(y,t)    \right| \ \le \    \frac{\rho C_{\la,\eta,\al,d}}{|y|^{d-1-\al}+1} \ , \quad y\in\mathbb{Z}^d \ ,
\ee
for some constant $C_{\la,\eta,\al,d}$ depending only on $\la,\eta,\al,d$. By a similar argument we also have that
\be \label{AH4}
\sup_{0<t\le T} \left|    D_{a_3}\om^{0,{\rm approx}}_{\ve,m,h,L}(y,t)    \right| \ \le \     \frac{\rho C_{\la,\eta,\al,d}}{|y|^{d-1-\al}+1} + \frac{\rho C_{\la,\eta,\al,d}}{|x-y|^{d-1-\al}+1}  \ .
\ee
Observe now from (\ref{T3}) that  for any $x\in\mathbb{Z}^d$, 
\be \label{AI4}
 \left\|    \na[-\De+m^2]^{-1}\na^*a^{\rm approx}_n(\cdot,t,T)    \right\|_{p,\tau_xw_{\al,p}} \ \le \    \exp\left[      \frac{\kappa_{\al,p}(1-\la+\eta)(T-t)}{2}            \right]\kappa_{\al,p}\|a\|_{p,\tau_xw_{\al,p}}  \ ,
\ee
where $\tau_xw_{\al,p}$ is the translation of the weight function $w_{\al,p}$ by $x$.  Taking $a=-\rho h_{x,\nu}$ in (\ref{AI4}) we then have that
\begin{multline} \label{AJ4}
\frac{1}{2}\int_0^T dt \ e^{-(T-t)/2}  \left\|    \na[-\De+m^2]^{-1}\na^*a^{\rm approx}_n(\cdot,t,T)    \right\|_{p,\tau_xw_{\al,p}} \\
 \le \ \frac{\rho\|h_{x,\nu}\|_{p,\tau_xw_{\al,p}}}{\la-\eta-(1-1/\kappa_{p,\al})} \ .
\end{multline}
It follows from (\ref{AC4}), (\ref{AJ4}) that
\be \label{AK4}
\frac{1}{2}\int_0^T dt \ e^{-(T-t)/2}  \left|    \na[-\De+m^2]^{-1}\na^*a^{\rm approx}_n(y,t,T)    \right| \ \le \  \frac{\rho C_{\la,\eta,\al,d}}{|x-y|^{d-1-\al}+1}  \ .
\ee
We then conclude from  (\ref{AB1}), (\ref{Z4}), (\ref{AG4}), (\ref{AH4}), (\ref{AK4})  that
\begin{multline} \label{AL4}
\left|      D^{3,{\rm approx}}_{a_1,a_2,a_3} q_{\ve,m,L}(h(\cdot))               \right|  \le \\  2M_\eta\sum_{y\in\mathbb{Z}^d} \frac{|\rho|^3 C_{\la,\eta,\al,d}^3}{\{|x-y|^{d-1-\al}+1\}\{|y|^{d-1-\al}+1\}^2} \
\le \ \frac{C_dM_\eta|\rho|^3 C_{\la,\eta,\al,d}^3}{|x|^{d-1-\al}+1} \ ,
\end{multline}
where the constant $C_d$ depends only on $d$. 

We have from (\ref{AF1}), (\ref{R4})-(\ref{W4}), (\ref{AA4}), (\ref{AL4}) that
\begin{multline} \label{AM4}
{\rm cov}_\ve\left\{   \exp\left[   \frac{\rho \phi(x) }{\ve}  \right] ,  \  \exp\left[   -\frac{ \rho\phi(0)}{\ve}  \right] \ \ \right\} \\
= \ \left\langle \  \exp\left[   -\frac{\rho \phi(0) }{\ve}  \right] \ \ \right\rangle_\ve^2\left\{\exp\left[       \frac{ -\ve^{-1}\rho^2\langle \phi(x)\phi(0)\rangle_\ve   +\rho^3{\rm Error}_\ve (x)}{\ve}                  \right] 
-1\right\} \ ,
\end{multline}
where $|{\rm Error}_\ve (x)|\le C/|x|^{d-1-\al}$ at large $|x|$.  From Theorem 1.1 of \cite{cf} we have that for some $\al>0$, 
\be \label{AN4}
\left|\ve^{-1}\langle \phi(x)\phi(0)\rangle_\ve  -G_{\mathbf{a}_{\ve,{\rm hom}}}(x)  \right| \ \le \ \frac{C}{|x|^{d-2+\al}} \ , \quad {\rm for \ }  |x|\ge 1 \ ,
\ee
where the constant $C$ is independent of $\ve>0$. Equation (\ref{AN1}) is a consequence of (\ref{AM4}), (\ref{AN4}). 

In the case $d=2$ we may derive equation (\ref{AO1}) in a similar way by using Theorem 1.3. Thus we choose $h=-\rho h_{x,\nu,L}+\rho h_{0,\nu,L}$ in (\ref{AC1}) and note that for any $p>2$ one has
$\limsup_{L\ra\infty,\nu\ra0}\| h_{0,\nu,L}\|_{p,L}\le C_p$,  where the constant $C_p$ depends only on $p$.  We may then take the limits $L\ra\infty,m\ra0,\nu\ra0$ as in the case $d\ge 3$ to obtain (\ref{AO1}) from (\ref{AC1}). 
\end{proof}
\begin{rem}
It is possible to obtain a bound on the covariance in (\ref{AN1}), which holds for any $\la>0$, provided $\rho\in\mathbb{R}$. This follows from the first identity in (\ref{S4})  by using the discrete Aronson estimate \cite{aronson, gos} for Green's functions for parabolic equations. The utility  of the Aronson inequality in the context of bounding correlation functions in Euclidean field theory was first observed in \cite{ns1}. 
\end{rem}

\end{document}